\documentclass[10pt]{article}

\usepackage{PRIMEarxiv}
\usepackage[utf8]{inputenc}
\usepackage[T1]{fontenc}

\usepackage{amsmath,amssymb,amsfonts,mathtools,amsthm}
\usepackage{graphicx}
\usepackage{subcaption}
\usepackage{booktabs,multirow,array,tabularx,longtable}
\usepackage{enumitem}
\usepackage{xcolor}
\usepackage{tcolorbox}
\usepackage{comment}
\usepackage{pgfplots}
\usepackage{float}
\usepackage{placeins}
\usepackage{nicefrac}
\usepackage{microtype}
\usepackage{url}
\usepackage{cite}
\usepackage{hyperref}
\usepackage[pageref]{backref}
\usepackage{makeidx}

\graphicspath{{./}{media/}}

\numberwithin{equation}{section}
\setlist[itemize]{nosep,leftmargin=*,topsep=0pt,partopsep=0pt}

\pgfplotsset{compat=newest}

\newtheorem{theorem}{Theorem}[section]

\newtheorem{proposition}[theorem]{Proposition}

\theoremstyle{definition}
\newtheorem{definition}[theorem]{Definition}

\newtheorem{problem}[theorem]{Problem}
\newtheorem{example}[theorem]{Example}

\theoremstyle{remark}
\newtheorem{remark}[theorem]{Remark}

\newcounter{myremark}[section]
\renewcommand{\themyremark}{\thesection.\arabic{myremark}}

\newtcolorbox{opinion}{
  colback=gray!5,
  colframe=black!40,
  boxrule=0.5pt,
  arc=2pt,
  left=6pt,
  right=6pt,
  top=4pt,
  bottom=4pt,
  fontupper=\small\sffamily,
}

\renewcommand*{\backrefalt}[4]{%
\ifcase #1 %
(Not cited)%
\or
(Cited on p.~#2)%
\else
(Cited on pp.~#2)%
\fi
}

\makeindex

\title{LLM Semantic Signaling Game and Mechanism Design: \\
Systematic Blindness, Awareness Shaping, and Mindset Dynamics}
\author{Quanyan Zhu\\
\small Department of Electrical and Computer Engineering\\
\small New York University Tandon School of Engineering\\
\small Brooklyn, NY 11201, USA\\
\small \texttt{qz494@nyu.edu}}

\begin{document}
\maketitle

\begin{abstract}
Large language models (LLMs) increasingly mediate strategic interaction
through natural language. This paper develops a semantic signaling game
for LLM-mediated communication in which a sender selects a high-level
semantic control, an LLM maps that control into a stochastic message, and
a receiver evaluates the message through an awareness-dependent scoring
rule. Receiver awareness is modeled as a type that determines which
linguistic features are perceived, weighted, and used for inference. The
framework connects prompt-based control, statistical detection, and
game-theoretic equilibrium analysis. We establish Gaussian
approximations for aggregate message scores, derive likelihood-ratio
decision rules, and characterize Perfect Bayesian Nash equilibria and
incentive-compatible control selection. We then study mechanism-design
questions in which interventions reshape receiver awareness, penalize
deceptive controls, or reweight populations to induce benign
pooling. Numerical experiments validate the Gaussian approximation,
show awareness-ordering effects, illustrate mindset dynamics, and
quantify benign-pooling design levers in a phishing case study. The
model explains systematic
blindness as a feature-space
limitation and provides a basis for analyzing mindset dynamics in
AI-enabled phishing and other adversarial language-mediated systems.
\end{abstract}


\section{Introduction}\label{sec:introduction}

Large language models (LLMs) are increasingly deployed as interactive
agents that communicate, reason, negotiate, and coordinate through
natural language. These systems are no longer passive text generators:
they participate in workflows, act on behalf of users, and interact with
other human or artificial agents. As a result, language becomes both a
medium of communication and a strategic control variable. A prompt,
instruction, or conversational framing can shape not only the surface
form of a generated message but also the receiver's inference and
subsequent action.

This shift changes the modeling problem. Classical signaling games
provide a natural starting point because they study communication under
asymmetric information. A sender has a private type, chooses a signal,
and a receiver interprets that signal before acting. However, standard
signaling models usually represent messages as abstract symbols or
low-dimensional numerical signals. They do not explicitly capture
token-level language generation, semantic control, or heterogeneous
perception of linguistic evidence. These limitations become especially
important in LLM-mediated systems, where messages are stochastic,
high-dimensional, and shaped by prompts.

The central difficulty is that an LLM-generated message is neither a
fixed action nor a fully controlled signal. It is a random token sequence
drawn from a distribution induced by a semantic input. The sender can
choose the high-level intent, framing, tone, or persuasive style of the
message, but the realized output is generated by the language model. The
receiver then observes the realized message through a limited
representation: some features are salient, others are ignored, and the
same message can lead to different scores for receivers with different
awareness levels. A useful theory must therefore connect semantic
control, stochastic generation, feature-based perception, and strategic
incentives in one model.

This difficulty is sharper because the receiver does not observe the
sender's semantic control directly. The receiver observes
only the generated message and must infer intent from linguistic evidence
that may be incomplete, noisy, or strategically chosen. A sender can
therefore shape beliefs without deterministically choosing every token.
Conversely, a receiver may respond rationally to the features it observes
while still missing the evidence that would distinguish benign and
malicious intent. The resulting interaction combines a stochastic
communication channel with a game of asymmetric information.

Phishing illustrates this need. Generative AI lowers the cost of
creating personalized, fluent, and adaptive deceptive messages. A
malicious sender may choose a semantic control that emphasizes urgency,
authority, or routine administrative framing, while a receiver evaluates
the message using only a limited set of cues. Some receivers notice
credential requests, suspicious links, or urgency markers; others miss
these features entirely. Thus, susceptibility is not simply a matter of
preference or random noise. It can arise from \emph{systematic
blindness}: a receiver's effective feature space may exclude cues that
are statistically informative about deception.

This paper develops a semantic signaling framework for this setting. We
model an LLM as a stochastic channel from semantic controls to token
sequences. The sender chooses a control based on its private type, and
the receiver maps the generated message into an awareness-dependent
score. This score is then used for belief updating and action selection.
The formulation makes it possible to connect three levels of analysis:
language generation, statistical detection, and equilibrium incentives.

The model also clarifies the role of receiver awareness. Awareness is
not treated as a scalar accuracy parameter alone. Instead, it determines
which token-level or semantic features are available for scoring. A
low-awareness receiver may fail to detect deception because the relevant
features are not represented, not because the receiver applies a weak
threshold to a correct statistic. This distinction matters for defense:
training, warnings, interface design, and automated support can change
the receiver's representation rather than merely shifting a decision
threshold.

The distinction also changes how interventions should be modeled. A
defense that only raises a threshold may reduce false acceptance, but it
does not teach the receiver to perceive previously invisible evidence. A
defense that expands the feature set changes the statistic itself. In an
LLM-mediated environment, both effects matter because a strategic sender
can adapt the semantic control to the receiver population. The paper
therefore treats awareness, thresholds, and sender incentives as coupled
objects rather than separate components.

The paper then studies mechanism design for awareness shaping. In a
defensive setting, one may seek interventions that make benign semantic
controls attractive for all sender types, increase the cost of deceptive
controls, or alter the receiver population so that deceptive controls
are less likely to be accepted. These interventions can be interpreted
as changes in detection thresholds, control costs, feature sets, or the
effective composition of awareness types. The resulting model provides a
language-based view of incentive design in agentic AI security.

The case study focuses on AI-enabled phishing, but the modeling idea is
broader. Many LLM-mediated interactions have the same structure: a
strategic party uses language to shape another party's belief or action,
and the receiver responds through a bounded and context-dependent
representation of the message. Examples include persuasion, automated
customer support, moderation-sensitive communication, misinformation,
and multi-agent coordination. In each case, semantic controls affect a
distribution over messages, while receiver awareness determines which
features of those messages enter the decision process.

\paragraph{Contributions.}
The paper makes several contributions. First, it formulates an LLM
semantic signaling game in which semantic controls induce stochastic
token sequences and awareness-dependent score distributions. Second, it
models receiver awareness as a type that governs token-level feature
extraction, which enables a formal treatment of heterogeneous detection
and systematic blindness. Third, it derives Gaussian approximations for
aggregate message scores and uses them to obtain likelihood-ratio and
threshold decision rules. Fourth, it characterizes Perfect Bayesian Nash
equilibria and incentive-compatible control selection for direct
generative mechanisms. Finally, it develops conditions for benign
pooling, population shaping, and mindset dynamics, and supports these
conditions with reproducible phishing experiments that validate the score
approximation, compare receiver-awareness types, and quantify
mechanism-design levers.

\paragraph{Organization.}
Section~\ref{sec:related-work} reviews related work. Section~\ref{sec:llm-game}
introduces the semantic signaling model, including stochastic message
generation and awareness-dependent scoring. Section~\ref{sec:pbne}
develops the equilibrium and incentive-compatibility analysis.
Section~\ref{sec:case} presents a numerical phishing case study.
Section~\ref{sec:conclusion} concludes.

\section{Related Work}\label{sec:related-work}

\paragraph{Signaling, cheap talk, and equilibrium refinements.}
The paper builds on classical models of signaling and strategic
communication. Spence's signaling model formalizes how costly signals
can transmit information under asymmetric information
\cite{spence1973job}. Crawford and Sobel's cheap-talk model studies
strategic communication when messages are payoff-relevant but not
directly binding \cite{crawford1982strategic}, and cheap-talk theory is
surveyed in \cite{farrell1996cheap}. Refinements and equilibrium
selection in signaling games are developed in
\cite{cho1987signaling,banks1987equilibrium}. Deception and lying in
games are surveyed and extended in \cite{sobel2020lying}. Our framework
keeps the asymmetric-information structure of these models but replaces
abstract messages with LLM-generated token sequences controlled by
semantic prompts.

\paragraph{Information design and persuasion.}
The mechanism-design part of the paper is related to information design,
Bayesian persuasion, and strategic disclosure. Bayesian persuasion shows
how an informed or committed sender can shape receiver beliefs by
choosing an information structure \cite{kamenica2011bayesian}; broader
unified treatments of information design appear in
\cite{bergemann2019information}. Work on overt persuasion and covert
signaling in cyber settings further emphasizes that the value of
communication depends on both information and observability
\cite{Li2023PriceTransparency}. In contrast with standard persuasion
models, the present model does not assume that the sender directly
selects a signal structure observed by a fully specified Bayesian
receiver. Instead, the sender selects semantic controls that induce a
distribution over natural-language messages, and receivers process those
messages through awareness-dependent feature maps.

\paragraph{Language, meaning, and conventions in games.}
Language has long been recognized as more than a neutral carrier of
signals. Conventions give messages meaning in coordination problems
\cite{lewis1969convention}, and evolutionary models show how meanings
can emerge through repeated interaction \cite{demichelis2008language}.
Interactive epistemology studies how agents reason about knowledge and
beliefs in games \cite{aumann1999interactive}. These works motivate a
view of language as a structured object with social and epistemic
content. The LLM setting adds a new layer: meanings are implemented
through probabilistic token generation, and the receiver's awareness
determines which parts of the linguistic object become decision-relevant.

\paragraph{Generative language models and agentic systems.}
The technical basis for modern LLMs includes the transformer
architecture \cite{vaswani2017attention}, few-shot language modeling
\cite{brown2020language}, instruction-following through human feedback
\cite{ouyang2022training}, and the broader foundation-model paradigm
\cite{bommasani2021opportunities}. Recent work studies LLMs as
strategic or interactive agents in games and agentic workflows.
Repeated-game behavior with LLM agents is examined in
\cite{akata2025playing}, and broader connections between LLMs and games
are surveyed in \cite{gallotta2024llm_games}. Generative agents show how
LLMs can support memory, planning, and social interaction in simulated
environments \cite{park2023generative}, while agentic AI workflows and
incentive-compatible teaming are studied in \cite{yang2026agentic_ai}.
Related work on LLM-based equilibrium reasoning and agentic AI security
appears in
\cite{zhu2025llm_nash,zhu2025llm,zhu2025generative,zhu2025cybersecurity}.
Our contribution is an explicit signaling-game model in which LLM
generation appears as a stochastic channel governed by semantic control.

\paragraph{Cyber deception, phishing, and cognitive security.}
Cybersecurity applications motivate the model because adversaries often
communicate strategically under asymmetric information. Game-theoretic
models of cyber deception and evidence-based signaling are developed in
\cite{Zhang2018HypothesisTestingGame,Pawlick2019LeakyDeception,
Hu2024GameTheoreticNP}. Phishing and one-to-many deception are studied
in \cite{pawlick2017phishing}, and social phishing demonstrates how
contextual personalization can increase attack effectiveness
\cite{jagatic2007social}. Attention enhancement and cognitive defenses
for phishing prevention are developed in
\cite{huang2022advert,huang2023cognitive}. Our work differs by modeling
the linguistic channel directly: phishing is not only an action chosen by
an attacker, but a distribution over generated messages shaped by
semantic controls and filtered through receiver awareness.

\paragraph{Statistical detection, awareness, and bounded rationality.}
The receiver's decision problem is closely related to hypothesis testing,
information theory, and misspecified inference. Classical
information-theoretic tools are summarized in \cite{cover1999elements},
and misspecification in statistical inference is studied in
\cite{white1982maximum}. The martingale central limit theorem provides
the probabilistic basis for the Gaussian score approximation used here
\cite{hall2014martingale}. The notions of systematic blindness and
mindset dynamics connect these statistical ideas to bounded rationality
\cite{simon1957models}: decision quality depends not only on thresholds
or priors, but also on which linguistic features are represented.

\section{LLM Signaling Game with Semantic Control and Receiver Awareness}
\label{sec:llm-game}

We now formalize the semantic signaling model. The sender controls the
high-level meaning of an LLM-generated message through a semantic input,
while receivers interpret the realized token sequence through
awareness-dependent detection systems. The model captures persuasion,
screening, and deception effects that arise when receivers differ in the
features they can perceive and score.

The communication process follows the transformation pipeline

\[
\theta
\;\xrightarrow{\mu}\;
y
\;\xrightarrow{\text{LLM}}\;
m
\;\xrightarrow{\Psi_\alpha}\;
S_\alpha
\;\xrightarrow{\delta_\alpha}\;
a ,
\]

where $\theta$ denotes the sender's hidden type, $y$ is the semantic
control selected by the sender, $m$ is the generated message, $S_\alpha$
is the score assigned by a receiver with awareness level $\alpha$, and
$a$ is the receiver's action.

\subsection{Sender Types and Semantic Control}

Let $\Theta$ be a finite set of sender types. The sender possesses a
private type $\theta \in \Theta$, which is not directly observable by
the receiver. The receiver begins the interaction with a prior belief
$\rho \in \Delta(\Theta)$, where $\Delta(\Theta)$ denotes the set of
probability distributions over $\Theta$. 
The sender's type $\theta$ encodes latent characteristics that influence
the communication strategy. Depending on the application, $\theta$ may
represent the sender's intent (e.g., benign or adversarial), the
reliability of the underlying information source, ideological position,
or the presence of strategic manipulation. In adversarial environments,
such as phishing or misinformation campaigns, the type space may
distinguish between legitimate and malicious senders; for example,
$\Theta=\{\theta^{\mathrm{legit}},\theta^{\mathrm{phish}}\}$, where
$\theta^{\mathrm{legit}}$ represents a legitimate communicator and
$\theta^{\mathrm{phish}}$ a malicious actor.

A key feature of language-model-mediated communication is that the
sender does not directly choose the realized message $m\in V^L$.
Instead, the sender selects a \emph{semantic control} that guides the
generation process. We therefore introduce a control variable
$y \in \mathcal{Y}$, where $\mathcal{Y}$ denotes the set of feasible
semantic instructions. The variable $y$ can be interpreted as a prompt,
instruction, or high-level specification of the intended message.

The semantic control $y$ influences the distribution of generated
messages by shaping linguistic attributes such as tone, framing,
persuasive intensity, topical emphasis, and stylistic structure. In
practice, $y$ corresponds to a prompt engineering strategy used to
steer the output of a language model. 
This interpretation is particularly natural in phishing settings. A
malicious sender may select a control $y$ that induces messages
emphasizing urgency (e.g., ``your account will be suspended within
24 hours''), authority (e.g., ``this message is from the bank's security
department''), or required action (e.g., ``verify your account
immediately''). Thus, $y$ determines a family of message distributions
that amplify specific persuasive or deceptive cues.

Because the sender's type determines the intended semantic content, the
sender's strategy is a mapping from types to semantic controls.

\begin{definition}[Semantic policy]
A sender (pure) strategy is a mapping $\mu:\Theta\to\mathcal{Y}$ that
assigns a semantic control $y=\mu(\theta)$ to each type $\theta\in\Theta$.
\end{definition}

Given a semantic control $y$, the language model induces a probability
distribution over messages, denoted by $\mathbb{P}(\cdot \mid y)$ on
$V^L$. Consequently, the sender does not control individual tokens
directly; rather, the sender selects $y$ to influence the induced
distribution of messages and, through it, the statistical properties of
the score $S_\alpha$.

This separation between type and message is central to the model. The
private type $\theta$ determines the semantic intent $y=\mu(\theta)$,
while the language model maps $y$ into a stochastic realization
$m\sim\mathbb{P}(\cdot\mid y)$. Thus, strategic control operates at the
level of distributions rather than individual token realizations.

 \subsection{LLM Message Generation}

Let $V$ be a finite vocabulary. A message of length $L \in \mathbb{N}$ is
an element $m=(w_1,\dots,w_L)\in V^L$. We model the generation of a
message as a stochastic process on the measurable space
$(V^L, \mathcal{F})$, where $\mathcal{F}$ is the $\sigma$-algebra
generated by cylinder sets.

The language model generates tokens sequentially via an autoregressive
mechanism. Let $\mathcal{Y}$ denote the set of admissible semantic
controls. For each $y \in \mathcal{Y}$, the model induces a sequence of
stochastic kernels
\[
\pi_y(\cdot \mid h_k) : V^{k-1} \to \Delta(V),
\]
where $h_k=(w_1,\dots,w_{k-1})\in V^{k-1}$ denotes the history prior to
stage $k$, and $\Delta(V)$ denotes the set of probability measures on
$V$. The quantity $\pi_y(w \mid h_k)$ represents the conditional
probability of generating token $w \in V$ given history $h_k$ and
control $y$.

Let $\{W_k\}_{k=1}^L$ denote the generated token sequence on
$(\Omega,\mathcal{G},\mathbb{P}_y)$. The process is adapted to the
natural filtration $F_k=\sigma(W_1,\dots,W_k)$, and its conditional law
satisfies
\begin{equation}
\mathbb{P}_y\!\left(W_k = w \mid F_{k-1}\right)
=
\pi_y(w \mid W_1,\dots,W_{k-1}),
\quad w \in V.
\label{eq:llm_kernel_formal}
\end{equation}

Equation \eqref{eq:llm_kernel_formal} specifies a controlled
non-homogeneous Markov process whose transition probabilities are
parametrized by $y$ and depend on the full history. 
By the Ionescu–Tulcea theorem, the sequence of stochastic kernels
$\{\pi_y(\cdot \mid h_k)\}_{k=1}^L$ uniquely induces a probability
measure $\mathbb{P}_y$ on $(V^L,\mathcal{F})$. The probability of
generating a message $m=(w_1,\dots,w_L)$ under control $y$ is therefore
given by
\begin{equation}
\mathbb{P}_y(m)
=
\prod_{k=1}^L \pi_y(w_k \mid w_1,\dots,w_{k-1}).
\label{eq:llm_chain_rule_formal}
\end{equation}

Thus, each semantic control $y$ induces a probability measure
$\mathbb{P}_y \in \Delta(V^L)$ over the message space. 
The semantic control $y$ acts as a parameter of the stochastic kernels
and therefore determines the law of the entire process
$\{W_k\}_{k=1}^L$. In particular, $y$ influences the distribution of any
measurable functional of the generated message. Of central importance
is the score
$S_\alpha = \Psi_\alpha(m)$,
which can be written as a function of the process
$\{W_k\}_{k=1}^L$.

Consequently, the language model can be interpreted as a stochastic
channel mapping the control variable $y$ into a probability measure
$\mathbb{P}_y$ on $V^L$, and hence into a distribution of scores
$S_\alpha$. The sender does not directly select a realization
$m \in V^L$, but instead chooses $y$ to control the induced measure
$\mathbb{P}_y$.

\subsection{Receiver Awareness as Type}

Receivers differ in their ability to interpret and evaluate generated
messages. Some are highly sensitive to linguistic patterns indicative of
manipulation or risk, while others may fail to detect such signals.
We formalize this heterogeneity by introducing an \emph{awareness type}
$\alpha \in \mathcal{A}$, which parametrizes the receiver's information
processing and detection capability.

The variable $\alpha$ is interpreted as the private type of the receiver. It captures both
cognitive and algorithmic characteristics that determine how a message
is evaluated. For human receivers, $\alpha$ can encode prior experience,
training, or familiarity with common attack patterns. For automated
systems, $\alpha$ may represent the structure or sensitivity of a
detection algorithm.

To model population-level heterogeneity, we assume that the awareness
type is drawn from a distribution $F \in \Delta(\mathcal{A})$, and write
$\alpha \sim F$. The measure $F$ characterizes the composition of the
receiver population and determines the prevalence of different
evaluation behaviors.

\begin{example}[Keyword-based phishing detection]

A concrete and analytically tractable representation arises in the
context of phishing detection. Consider receivers who evaluate messages
by screening for the presence of suspicious keywords. Let
$K_\alpha \subset V$ denote a set of alert tokens associated with
awareness type $\alpha$. Upon receiving a message
$m=(w_1,\dots,w_L)$, the receiver inspects whether tokens from
$K_\alpha$ appear in the message.

Different awareness types correspond to different keyword sets, which
reflect the receiver's ability to recognize increasingly subtle and
context-dependent indicators of malicious intent. A natural
specification is to consider a nested family of sets
$\{K_\alpha\}_{\alpha \in \mathcal{A}}$ such that higher awareness
corresponds to a richer and more discriminative set of tokens.

A low-awareness receiver $\alpha_{\mathrm{low}}$ may rely only on the
most explicit credential-related cues,
\[
K_{\alpha_{\mathrm{low}}}
=
\{\text{``password''},\ \text{``login''},\ \text{``verify''},\ \text{``PIN''}\},
\]
which capture direct attempts to solicit sensitive information.

An intermediate-awareness receiver $\alpha_{\mathrm{mid}}$ augments this
set with operational, urgency, and security-related tokens,
\[
K_{\alpha_{\mathrm{mid}}}
=
K_{\alpha_{\mathrm{low}}}
\cup
\{\text{``account''},\ \text{``urgent''},\ \text{``security''},\ \text{``update''},\ \text{``alert''},\ \text{``suspended''}\},
\]
which capture common patterns used to induce time pressure and
legitimacy.

A highly aware receiver $\alpha_{\mathrm{high}}$ further expands the set
to include more subtle transactional, behavioral, and interaction cues,
\[
\begin{aligned}
K_{\alpha_{\mathrm{high}}}
=
K_{\alpha_{\mathrm{mid}}}
\cup
\{&
\text{``transfer''},\ \text{``wire''},\ \text{``payment''},\ \text{``confirm''},\\
&\text{``click''},\ \text{``link''},\ \text{``reset''},\ \text{``authentication''},\\
&\text{``invoice''},\ \text{``bank''},\ \text{``verify identity''}
\}.
\end{aligned}
\]

Thus, the family $\{K_\alpha\}$ satisfies a monotonicity property
$K_{\alpha_1} \subseteq K_{\alpha_2}$ whenever
$\alpha_1 \le \alpha_2$, reflecting that more aware receivers monitor a
strictly richer set of linguistic signals, including both explicit and
contextual indicators of phishing. 
Given a message $m$, a simple detection mechanism is the indicator
\(
\mathbf{1}\!\left\{\exists\, k \le L \text{ such that } w_k \in K_\alpha\right\},
\)
which triggers an alert whenever at least one suspicious token appears.
More generally, the receiver may assign heterogeneous weights to tokens
in $K_\alpha$, leading to the additive scoring representation that will be introduced
in Section~\ref{sec:scoring}.
\end{example}
 
\subsection{Receiver Scoring Systems}\label{sec:scoring}

Each receiver evaluates incoming messages using an internal scoring
mechanism. We represent this mechanism by a receiver-specific operator
$\Psi_\alpha: V^L \to \mathbb{R}$, which maps a message
$m \in V^L$ to a scalar signal $S_\alpha = \Psi_\alpha(m)$. The score
$S_\alpha$ represents the receiver's internal assessment of the message.
Depending on the application, this signal may encode moderation risk,
toxicity or safety indicators, ideological intensity, likelihood of
phishing, or more generally a measure of credibility or trustworthiness.

The mapping $\Psi_\alpha$ depends on the receiver's awareness type
$\alpha \in \mathcal{A}$. Thus, heterogeneity in awareness induces
heterogeneity in scoring rules across the population. In particular,
different receivers may emphasize different linguistic features when
evaluating the same message, leading to heterogeneous score realizations
for a fixed $m \in V^L$.

A tractable representation of $\Psi_\alpha$ is obtained by decomposing
the evaluation process at the token level.

\begin{definition}[Token scoring function]
For each $\alpha \in \mathcal{A}$, let
$\psi_\alpha: V \to \mathbb{R}$ be a measurable function that assigns a
score to each token. We assume that $\psi_\alpha$ is uniformly bounded,
i.e., there exists a constant $C>0$ such that
$|\psi_\alpha(w)| \le C$ for all $w \in V$ and all $\alpha \in \mathcal{A}$.
\end{definition}

Given a message $m=(w_1,\dots,w_L)$, the induced score is
\begin{equation}
S_\alpha(m)
=
\sum_{k=1}^L \psi_\alpha(w_k),
\label{eq:token_score}
\end{equation}
so that each token contributes an additive increment to the overall
evaluation.

\begin{example}[Connection to keyword-based detection.]

The token scoring representation generalizes the keyword-based detection
mechanism introduced in the previous subsection. In particular, for a
given awareness type $\alpha$ with keyword set $K_\alpha \subset V$, a
simple detection rule can be written as
\[
\psi_\alpha(w) = \mathbf{1}\{w \in K_\alpha\},
\]
in which case $S_\alpha(m)$ counts the number of suspicious tokens in
the message. More generally, different tokens may be assigned
heterogeneous weights, allowing for graded sensitivity to different
linguistic cues.

For example, a highly aware receiver $\alpha_{\mathrm{high}}$ may assign
large weights to tokens such as ``password'', ``verify'', ``account'',
``urgent'', or ``transfer'', while a less aware receiver
$\alpha_{\mathrm{low}}$ assigns smaller weights or ignores some of these
tokens altogether. Consequently, for the same message $m$, one may have
$S_{\alpha_{\mathrm{high}}}(m) \gg S_{\alpha_{\mathrm{low}}}(m)$,
indicating stronger detection capability.
\end{example}

The function $\psi_\alpha$ can be interpreted as a feature extraction or
filtering mechanism that maps linguistic inputs to risk-relevant
signals. Receivers with higher awareness apply more sensitive filters 
that amplify suspicious patterns in the text, whereas receivers with
lower awareness apply weaker filters and may fail to detect such
patterns. This formulation captures heterogeneity in both human
judgment and algorithmic classification systems within a unified
mathematical framework.

\subsection{Gaussian Approximation of Message Scores}

Recall from \eqref{eq:token_score} that the token-level score is
$X_{k,\alpha} = \psi_\alpha(W_k)$, where $W_k$ denotes the token
generated at stage $k$. The cumulative score is
\[
S_{\alpha,L} = \sum_{k=1}^{L} X_{k,\alpha}.
\]

Fix a semantic control $y \in \mathcal{Y}$ and let
$(\Omega,\mathcal{F},\mathbb{P}_y)$ denote the probability space
induced by the language model under $y$. The token sequence
$\{W_k\}_{k=1}^L$ is adapted to the natural filtration
$F_k = \sigma(W_1,\dots,W_k)$.

\paragraph{Martingale decomposition.}

Define the predictable process
$m_{k,\alpha} := \mathbb{E}_y[X_{k,\alpha} \mid F_{k-1}]$
and the martingale difference sequence
$Y_{k,\alpha} := X_{k,\alpha} - m_{k,\alpha}$. Then
$\{Y_{k,\alpha}, F_k\}$ is a martingale difference sequence satisfying
$\mathbb{E}_y[Y_{k,\alpha} \mid F_{k-1}] = 0$. Consequently, the cumulative score admits the decomposition
\begin{equation}
S_{\alpha,L}
=
\sum_{k=1}^{L} m_{k,\alpha}
+
M_{\alpha,L},
\label{eq:score_decomposition}
\end{equation}
where $M_{\alpha,L} := \sum_{k=1}^{L} Y_{k,\alpha}$ is a martingale.

\paragraph{Predictable quadratic variation.}

Let $\sigma_{k,\alpha}^2$ denote
$\mathbb{E}_y[Y_{k,\alpha}^2 \mid F_{k-1}]$. The predictable quadratic
variation is
\begin{equation}
\langle M_\alpha \rangle_L
=
\sum_{k=1}^{L} \sigma_{k,\alpha}^2.
\label{eq:predictable_variation}
\end{equation}

\paragraph{Regularity conditions.}

We impose the following conditions:

\begin{enumerate}
\item[(A1)] \textbf{Bounded increments:} There exists $C>0$ such that
$|\psi_\alpha(w)| \le C$ for all $w \in V$.

\item[(A2)] \textbf{Variance growth:}
$\langle M_\alpha \rangle_L \to \infty$ in probability as $L \to \infty$.

\item[(A3)] \textbf{Conditional Lindeberg condition:} For every $\varepsilon > 0$,
\[
\frac{1}{\langle M_\alpha \rangle_L}
\sum_{k=1}^L
\mathbb{E}_y\!\left[
Y_{k,\alpha}^2 \mathbf{1}\{|Y_{k,\alpha}| > \varepsilon \sqrt{\langle M_\alpha \rangle_L}\}
\,\middle|\, F_{k-1}
\right]
\;\xrightarrow{P}\; 0.
\]
\end{enumerate}

Condition (A3) is automatically satisfied under (A1) since the increments
are uniformly bounded. Under conditions (A1)--(A3), a martingale central limit theorem (e.g., see \cite{hall2014martingale}) implies
\begin{equation}
\frac{M_{\alpha,L}}
{\sqrt{\langle M_\alpha \rangle_L}}
\;\Rightarrow\;
N(0,1).
\label{eq:martingale_clt}
\end{equation}

\paragraph{Asymptotic normality of the score.}

Let $\bar{s}_\alpha(y) := \mathbb{E}_y[S_{\alpha,L}]$. From
\eqref{eq:score_decomposition}, we have
$S_{\alpha,L} = \bar{s}_\alpha(y) + M_{\alpha,L}$. If, in addition,
\begin{equation}
\mathrm{Var}_y(S_{\alpha,L})
\sim
\mathbb{E}_y[\langle M_\alpha \rangle_L],
\label{eq:variance_equivalence}
\end{equation}
then
\begin{equation}
\frac{S_{\alpha,L} - \bar{s}_\alpha(y)}
{\sqrt{\mathrm{Var}_y(S_{\alpha,L})}}
\;\Rightarrow\;
N(0,1).
\label{eq:score_clt}
\end{equation}

\begin{proposition}[Gaussian approximation of message scores]
\label{prop:gaussian_scores}
\label{prop:score_clt}

Under conditions (A1)--(A3) and \eqref{eq:variance_equivalence}, the
cumulative score $S_{\alpha,L}$ satisfies
\begin{equation}
\frac{S_{\alpha,L} - \bar{s}_\alpha(y)}
{\sigma_\alpha(y)}
\;\Rightarrow\;
N(0,1),
\end{equation}
where $\bar{s}_\alpha(y) = \mathbb{E}_y[S_{\alpha,L}]$ and
$\sigma_\alpha^2(y) = \mathrm{Var}_y(S_{\alpha,L})$. 
Consequently, for sufficiently large $L$,
\begin{equation}
S_{\alpha,L}
\approx
N\!\left(\bar{s}_\alpha(y),\, \sigma_\alpha^2(y)\right).
\label{eq:gaussian_approximation}
\end{equation}

\end{proposition}

Proposition~\ref{prop:gaussian_scores} shows that the cumulative score
induced by the language model behaves asymptotically as a Gaussian
signal whose mean and variance depend on the semantic control $y$ and
the receiver type $\alpha$. The sender, through the choice of $y$,
controls the location of the distribution, while the stochastic nature
of token generation determines its dispersion. This Gaussian structure
forms the basis for the receiver's inference problem and the equilibrium
analysis developed in subsequent sections.

\subsection{Receiver Beliefs, Utility, and Decision Rules}

After observing the generated message $m$ through its score
$S_\alpha=\Psi_\alpha(m)$, the receiver updates its belief about the
sender's type and selects an action. 
Let $\rho \in \Delta(\Theta)$ denote the prior receiver and sender
types. For a fixed semantic control $y \in \mathcal{Y}$ and awareness
level $\alpha$, the posterior belief is given by Bayes' rule:
\begin{equation}
\rho(\theta \mid s, y, \alpha)
=
\frac{\rho(\theta)\, f_{\alpha,\theta}(s \mid y)}
{\sum_{\theta' \in \Theta} \rho(\theta')\, f_{\alpha,\theta'}(s \mid y)},
\label{eq:posterior_explicit}
\end{equation}
where $s=S_\alpha$ and $f_{\alpha,\theta}(\cdot \mid y)$ denotes the
density of the score under type $\theta$. 
Under Proposition~\ref{prop:gaussian_scores}, the score satisfies the
Gaussian approximation
\begin{equation}
S_{\alpha,L} \mid (\theta,y)
\approx
N\!\left(\bar{s}_{\alpha,\theta}(y),\, \sigma_{\alpha,\theta}^2(y)\right),
\label{eq:score_distribution}
\end{equation}
and hence the likelihood function is
\begin{equation}
f_{\alpha,\theta}(s \mid y)
=
\frac{1}{\sqrt{2\pi \sigma_{\alpha,\theta}^2(y)}}
\exp\!\left(
-\frac{\big(s-\bar{s}_{\alpha,\theta}(y)\big)^2}
{2\sigma_{\alpha,\theta}^2(y)}
\right).
\label{eq:likelihood}
\end{equation}

Let $U_R : A \times \Theta \to \mathbb{R}$ denote the receiver's utility
function. Given posterior beliefs \eqref{eq:posterior_explicit}, the
expected utility of action $a \in A$ is
\begin{equation}
\mathbb{E} U_R(a \mid s, y, \alpha)
=
\sum_{\theta \in \Theta}
\rho(\theta \mid s, y, \alpha)\, U_R(a,\theta).
\label{eq:expected_utility}
\end{equation}

The receiver's strategy is a measurable mapping
$\delta_\alpha : \mathbb{R} \to A$ such that
\begin{equation}
\delta_\alpha(s)
\in
\arg\max_{a\in A}
\mathbb{E} U_R(a \mid s, y, \alpha).
\label{eq:optimal_decision}
\end{equation}

Consider $A=\{\text{allow},\text{block}\}$ and
$\Theta=\{\theta^{\text{legit}},\theta^{\text{mal}}\}$. Define the
likelihood ratio
\begin{equation}
\Lambda(s)
:=
\frac{f_{\alpha,\theta^{\text{mal}}}(s \mid y)}
{f_{\alpha,\theta^{\text{legit}}}(s \mid y)}.
\label{eq:likelihood_ratio}
\end{equation}

Using \eqref{eq:posterior_explicit}, the posterior odds ratio satisfies
\begin{equation}
\frac{\rho(\theta^{\text{mal}} \mid s,y,\alpha)}
{\rho(\theta^{\text{legit}} \mid s,y,\alpha)}
=
\frac{\rho(\theta^{\text{mal}})}{\rho(\theta^{\text{legit}})}
\cdot
\Lambda(s).
\label{eq:posterior_odds}
\end{equation}

\begin{proposition}[Threshold optimality of the receiver]
\label{prop:threshold_policy}

Assume $A=\{\text{allow},\text{block}\}$ and that
$U_R(\text{allow},\theta^{\text{legit}})
>
U_R(\text{block},\theta^{\text{legit}})$ and
$U_R(\text{block},\theta^{\text{mal}})
>
U_R(\text{allow},\theta^{\text{mal}})$. 
Then, the optimal decision rule \eqref{eq:optimal_decision} is a
likelihood-ratio test: there exists $\kappa>0$ such that
\begin{equation}
\delta_\alpha(s)
=
\begin{cases}
\text{block}, & \Lambda(s) > \kappa,\\
\text{allow}, & \Lambda(s) \le \kappa,
\end{cases}
\label{eq:lr_rule}
\end{equation}
where
\begin{equation}
\kappa
=
\frac{\rho(\theta^{\text{legit}})}{\rho(\theta^{\text{mal}})}
\cdot
\frac{U_R(\text{allow},\theta^{\text{legit}})
-
U_R(\text{block},\theta^{\text{legit}})}
{U_R(\text{block},\theta^{\text{mal}})
-
U_R(\text{allow},\theta^{\text{mal}})}.
\label{eq:kappa}
\end{equation}

\end{proposition}

\begin{proof}

The receiver prefers \text{block} over \text{allow} if and only if
\begin{equation}
\mathbb{E}U_R(\text{block} \mid s,y,\alpha)
\ge
\mathbb{E}U_R(\text{allow} \mid s,y,\alpha).
\label{eq:decision_condition}
\end{equation}

Using \eqref{eq:expected_utility}, this is equivalent to
\begin{align}
&\rho(\theta^{\text{mal}} \mid s,y,\alpha)
\big(U_R(\text{block},\theta^{\text{mal}})
-
U_R(\text{allow},\theta^{\text{mal}})\big)
\nonumber\\
&\quad\ge
\rho(\theta^{\text{legit}} \mid s,y,\alpha)
\big(U_R(\text{allow},\theta^{\text{legit}})
-
U_R(\text{block},\theta^{\text{legit}})\big).
\end{align} 
Rearranging and applying \eqref{eq:posterior_odds} yields
\(
\Lambda(s) \ge \kappa,
\)
which proves \eqref{eq:lr_rule}.
\end{proof}

Under \eqref{eq:score_distribution}, if
$\sigma_{\alpha,\theta^{\text{mal}}}^2(y)
=
\sigma_{\alpha,\theta^{\text{legit}}}^2(y)$, then
$\log \Lambda(s)$ is affine in $s$, and hence $\Lambda(s)$ is monotone.
Thus there exists a threshold $\eta(\alpha)$ such that
\begin{equation}
\delta_\alpha(s)
=
\begin{cases}
\text{allow}, & s \le \eta(\alpha),\\
\text{block}, & s > \eta(\alpha).
\end{cases}
\label{eq:threshold_rule}
\end{equation}

Under the Gaussian approximation,
\begin{equation}
\mathbb{P}\!\left(S_\alpha \le \eta(\alpha)\mid y\right)
=
\Phi\!\left(
\frac{\eta(\alpha)-\bar{s}_\alpha(y)}
{\sigma_\alpha(y)}
\right),
\label{eq:acceptance_probability}
\end{equation}
where $\Phi$ denotes the standard normal cumulative distribution
function.

\subsection{Sender Payoff and Audience Matching}

The sender's payoff depends on its type and on which receivers accept
the generated message. Because receivers differ in their awareness
levels, the sender interacts with a heterogeneous population. 
Let $\theta \in \Theta$ denote the sender's type and
$\alpha \in \mathcal{A}$ the receiver's awareness level. Let
$F \in \Delta(\mathcal{A})$ denote the population distribution of
receiver types.

\paragraph{Sender valuation.}

When a receiver of awareness level $\alpha \in \mathcal{A}$ accepts a
message generated under semantic control $y \in \mathcal{Y}$ by a sender
of type $\theta \in \Theta$, the sender obtains utility
\begin{equation}
V_S(\theta,y,\alpha),
\label{eq:sender_value}
\end{equation}
representing the type-dependent benefit of a successful interaction. 
To model heterogeneous targeting, let
\begin{equation}
w(\alpha) \ge 0
\label{eq:target_weight}
\end{equation}
be a measurable weighting function over receiver types, capturing their
relative importance.

We adopt the structured specification
\begin{equation}
V_S(\theta,y,\alpha)
=
v_\theta(\alpha)
-
C_\theta(y),
\label{eq:sender_value_structured}
\end{equation}
where $v_\theta:\mathcal{A}\to\mathbb{R}_+$ denotes the value of engaging
a receiver of type $\alpha$, and $C_\theta:\mathcal{Y}\to\mathbb{R}_+$
represents the cost of implementing control $y$.

The function $v_\theta(\alpha)$ encodes targeting preferences and
differs across types. For $\theta^{\mathrm{ben}}$, it is typically
increasing in $\alpha$, reflecting higher value from more aware
receivers. For $\theta^{\mathrm{mal}}$, it is typically decreasing in
$\alpha$, reflecting a preference for less aware, more susceptible
targets. 
The cost $C_\theta(y)$ captures the effort or risk associated with
shaping the signal distribution, including tradeoffs between
persuasiveness and detectability.

The product $w(\alpha)\, v_\theta(\alpha)$ defines the effective weight
assigned to receivers of type $\alpha$. Thus, the sender's objective can
be interpreted as weighting acceptance probabilities by an
\emph{effective targeting measure} over the population.  

\paragraph{Acceptance probability.}

From \eqref{eq:threshold_rule} and \eqref{eq:acceptance_probability},
a receiver with awareness level $\alpha$ accepts the message if
$S_{\alpha,L} \le \eta(\alpha)$. Under the Gaussian approximation, the
acceptance probability conditional on $(\theta,y)$ is
\begin{equation}
\mathbb{P}\!\left(S_{\alpha,L} \le \eta(\alpha)
\mid \theta,y,\alpha\right)
=
\Phi\!\left(
\frac{\eta(\alpha)-\bar{s}_{\alpha,\theta}(y)}
{\sigma_{\alpha,\theta}(y)}
\right),
\label{eq:acceptance_prob_sender}
\end{equation}
where
\[
\bar{s}_{\alpha,\theta}(y)
=
\mathbb{E}\!\left[S_{\alpha,L}\mid \theta,y,\alpha\right],
\qquad
\sigma_{\alpha,\theta}^2(y)
=
\mathrm{Var}\!\left(S_{\alpha,L}\mid \theta,y,\alpha\right).
\]

\paragraph{Expected payoff.}

The sender's expected payoff, conditional on type $\theta$, is given by
\begin{equation}
U_S(\theta,y)
=
\int_{\mathcal{A}}
w(\alpha)\,
V_S(\theta,y,\alpha)\,
\mathbb{P}\!\left(S_{\alpha,L} \le \eta(\alpha)
\mid \theta,y,\alpha\right)
\, dF(\alpha).
\label{eq:sender_payoff}
\end{equation}

Substituting \eqref{eq:acceptance_prob_sender} yields
\begin{equation}
U_S(\theta,y)
=
\int_{\mathcal{A}}
w(\alpha)\,
V_S(\theta,y,\alpha)\,
\Phi\!\left(
\frac{\eta(\alpha)-\bar{s}_{\alpha,\theta}(y)}
{\sigma_{\alpha,\theta}(y)}
\right)
\, dF(\alpha).
\label{eq:sender_payoff_gaussian}
\end{equation}

\paragraph{Sender optimization problem.}

The sender selects a semantic control $y \in \mathcal{Y}$ to maximize
expected payoff. An optimal control $y^*$ satisfies
\begin{equation}
y^*
\in
\arg\max_{y \in \mathcal{Y}}
U_S(\theta,y),
\label{eq:sender_problem}
\end{equation}
for each $\theta \in \Theta$.  
We represent the sender's behavior by a mapping
$\mu:\Theta \to \mathcal{Y}$, where
\begin{equation}
\mu(\theta) := y^*
\label{eq:mu_mapping}
\end{equation}
denotes the semantic control selected by a sender of type $\theta$.

The function $\mu(\theta)$ is a control variable and should not be
confused with the mean of the score distribution. Rather, it induces the
distribution of scores through
\begin{equation}
\bar{s}_{\alpha,\theta}(\mu(\theta))
=
\mathbb{E}\!\left[S_{\alpha,L}\mid \theta,\mu(\theta),\alpha\right],
\label{eq:induced_mean}
\end{equation}
and similarly $\sigma_{\alpha,\theta}(\mu(\theta))$ determines the
dispersion.

Equation~\eqref{eq:sender_problem} defines a type-dependent optimization
problem in which the sender selects a control $\mu(\theta)$ to shape the
distribution of message scores. The functions
$\bar{s}_{\alpha,\theta}(\mu(\theta))$ and
$\sigma_{\alpha,\theta}(\mu(\theta))$ determine, respectively, the
location and dispersion of the score distribution perceived by receivers
of type $\alpha$. The sender uses $\mu(\theta)$ to trade off acceptance
across heterogeneous receiver types. The weighting function $w(\alpha)$
and valuation function $V_S(\theta,\mu(\theta),\alpha)$ encode audience
targeting preferences.

In adversarial settings, such as phishing, a malicious sender
$\theta^{\mathrm{mal}}$ selects $\mu(\theta)$ to increase acceptance
probability among less-aware receivers while avoiding more-aware ones.
This induces a systematic shift of the score
distribution toward regions below the thresholds $\eta(\alpha)$ for
targeted audiences. Consequently, the sender's problem can be
interpreted as an \emph{audience matching problem}: the control policy
$\mu(\theta)$ selects a signal distribution that maximizes expected
payoff under heterogeneous inference and decision rules.

\section{Perfect Bayesian Nash Equilibrium}
\label{sec:pbne}

We now define the equilibrium of the signaling game induced by the
controlled language generation model. A sender of type $\theta \in \Theta$ selects a semantic control
$y \in \mathcal{Y}$. Through the stochastic kernels
\eqref{eq:llm_kernel_formal}, this induces a probability measure
$\mathbb{P}_y$ on the message space $V^L$ via
\eqref{eq:llm_chain_rule_formal}, and hence a distribution over scores
\(
S_\alpha = \Psi_\alpha(m),
\) 
defined in \eqref{eq:token_score}.

The receiver of awareness type $\alpha \in \mathcal{A}$ observes the
realization $s \in \mathbb{R}$ of $S_\alpha$ and selects an action
$a \in A$ according to \eqref{eq:optimal_decision}. Beliefs are updated
using Bayes' rule \eqref{eq:posterior_explicit}, with likelihood given by
\eqref{eq:likelihood} under the Gaussian approximation
\eqref{eq:score_distribution}. 
The sender's payoff is given by \eqref{eq:sender_payoff_gaussian}. We represent the sender's strategy by a measurable mapping
\(
\mu : \Theta \to \mathcal{Y},
\) 
as defined in \eqref{eq:mu_mapping}.

\begin{definition}[Perfect Bayesian Nash Equilibrium]
A \emph{Perfect Bayesian Nash Equilibrium (PBNE)} consists of

\begin{enumerate}
\item[(i)] a sender strategy $\mu^* : \Theta \to \mathcal{Y}$,

\item[(ii)]  a receiver decision rule $\delta_\alpha^* : \mathbb{R} \to A$ for
each $\alpha \in \mathcal{A}$,

\item[(iii)]  and a belief system $\rho^*(\theta \mid s, \alpha)$,
\end{enumerate}

such that the following conditions hold.
\end{definition}

\paragraph{Sender optimality.}

For every $\theta \in \Theta$, the equilibrium control $\mu^*(\theta)$
solves the sender's problem \eqref{eq:sender_problem}, i.e.,
\begin{equation}
\mu^*(\theta)
\in
\arg\max_{y \in \mathcal{Y}}
U_S(\theta,y),
\label{eq:pbe_sender}
\end{equation}
where $U_S(\theta,y)$ is given by
\eqref{eq:sender_payoff_gaussian}. The optimization accounts for the
induced score distribution through
$\bar{s}_{\alpha,\theta}(y)$ and $\sigma_{\alpha,\theta}(y)$ defined in
\eqref{eq:score_distribution}.

\paragraph{Receiver optimality.}

For each $\alpha \in \mathcal{A}$ and $s \in \mathbb{R}$, the decision
rule $\delta_\alpha^*$ satisfies
\begin{equation}
\delta_\alpha^*(s)
\in
\arg\max_{a \in A}
\sum_{\theta \in \Theta}
\rho^*(\theta \mid s,\alpha)\, U_R(a,\theta),
\label{eq:pbe_receiver}
\end{equation}
where the posterior $\rho^*(\theta \mid s,\alpha)$ is given by
\eqref{eq:posterior_explicit} under the equilibrium strategy $\mu^*$.

\paragraph{Belief consistency.}

For all $s$ in the support of the equilibrium score distribution,
beliefs satisfy Bayes' rule:
\begin{equation}
\rho^*(\theta \mid s,\alpha)
=
\frac{\rho(\theta)\,
f_{\alpha,\theta}(s \mid \mu^*(\theta))}
{\sum_{\theta' \in \Theta}
\rho(\theta')\,
f_{\alpha,\theta'}(s \mid \mu^*(\theta'))},
\label{eq:pbe_bayes}
\end{equation}
where $f_{\alpha,\theta}(s \mid y)$ is defined in
\eqref{eq:likelihood}.

In equilibrium, the sender selects controls $\mu^*(\theta)$ that shape
the induced probability measure $\mathbb{P}_{\mu^*(\theta)}$ on
the messages and, consequently, the Gaussian score distribution
\eqref{eq:score_distribution}. The receiver performs statistical
inference based on $s$ and responds optimally. Under Proposition~\ref{prop:threshold_policy}, the receiver's strategy
reduces to a threshold rule \eqref{eq:threshold_rule}. Hence, the
equilibrium can be characterized as a fixed point between: (i) \emph{semantic control} $y = \mu(\theta)$, which determines the
distribution of scores through \eqref{eq:llm_chain_rule_formal} and
\eqref{eq:token_score}, and 
(ii) \emph{statistical detection}, in which receivers apply
likelihood-ratio tests based on \eqref{eq:likelihood}. 
This fixed point captures the interaction between controlled language
generation and heterogeneous inference across the receiver population.

\subsection{Direct Revelation with Generative Types}

We consider a finite type space $\Theta = \mathcal{Y} = \{\theta_0,\theta_1,\dots,\theta_K\}$, where each type $\theta \in \Theta$ corresponds to a canonical semantic control (prompt). A sender of true type $\theta$ selects a reported control $y \in \Theta$, and a strategy is a mapping $\mu : \Theta \to \Theta$. For each $\theta \in \Theta$, the semantic control induces a sequence of stochastic kernels $\{\pi_\theta(\cdot \mid h_k)\}_{k=1}^L$ as in \eqref{eq:llm_kernel_formal}, which uniquely define a probability measure $\mathbb{P}_\theta \in \Delta(V^L)$ via \eqref{eq:llm_chain_rule_formal}. Hence, each type $\theta$ determines a generative law over messages. With $\Theta = \mathcal{Y}$, reporting $y$ is equivalent to selecting the generative mechanism $\mathbb{P}_y$. A sender of type $\theta$ who reports $y$ induces the distribution $\mathbb{P}_y$ while retaining its intrinsic valuation $v_\theta(\cdot)$.

\begin{example}[Phishing with generative types]
\label{ex:phishing_generative_types_inline}

Consider a phishing detection setting in which a sender communicates through an LLM. In a classical formulation, the type space is $\Theta^{\text{orig}} = \{\text{legit}, \text{phish}\}$, which captures the sender’s intent, and the message (or semantic control) space $\mathcal{Y}^{\text{orig}}$ consists of prompts that can be used to generate messages. We assume that the legitimate type has a single baseline semantic control $y_0$ (truthful communication), while the phishing type has access to multiple semantic controls $y_1,y_2,y_3$. For instance, $y_1$ may correspond to an urgent credential request, $y_2$ to an invoice or payment scam, and $y_3$ to a security alert impersonation. In this representation, a sender of type $\theta \in \Theta^{\text{orig}}$ selects a prompt $y \in \mathcal{Y}^{\text{orig}}$, so the choice of $y$ depends on $\theta$ but the type itself does not distinguish among different phishing strategies.

In contrast, we construct a refined representation in which both the type space and the message space are given by a finite set of canonical semantic controls, $\Theta = \mathcal{Y} = \{\theta_0,\theta_1,\theta_2,\theta_3\}$. Each element $\theta_i$ corresponds to a prompt template, with $\theta_0$ representing legitimate communication and $\theta_1,\theta_2,\theta_3$ representing distinct phishing strategies such as credential requests, invoice scams, and security alerts. We identify each type $\theta_i$ with a prompt $y_i$, so that $\theta_i \leftrightarrow y_i$. In this refined model, the heterogeneity of phishing behavior is incorporated directly into the type space, and reporting $y \in \mathcal{Y}$ is equivalent to selecting a generative mechanism indexed by $\theta_i$.

\end{example} 

Let $\mathcal{A} = \{\alpha_1,\dots,\alpha_N\}$ be the finite receiver type space with distribution $F(\alpha_n)=p_n$. Under the Gaussian approximation \eqref{eq:score_distribution}, the acceptance probability for receiver type $\alpha_n$ is
\begin{equation}
P_n(\theta,y)
=
\Phi\!\left(
\frac{\eta(\alpha_n) - \bar{s}_{\alpha_n,\theta}(y)}
{\sigma_{\alpha_n,\theta}(y)}
\right).
\label{eq:acceptance_discrete}
\end{equation}
 
 The expected payoff of type $\theta$ under control $y$ is
\begin{equation}
U_S(\theta,y)
=
\sum_{n=1}^N
p_n\, w(\alpha_n)\, v_\theta(\alpha_n)\, P_n(\theta,y).
\label{eq:sender_payoff_discrete}
\end{equation}
Define the weights $\tilde{w}_{\theta,n} := p_n\, w(\alpha_n)\, v_\theta(\alpha_n)$. Then
\begin{equation}
U_S(\theta,y)
=
\sum_{n=1}^N
\tilde{w}_{\theta,n}\, P_n(\theta,y),
\label{eq:sender_payoff_compact}
\end{equation}
which exhibits $U_S(\theta,y)$ as a weighted aggregation of acceptance probabilities.

Define $U^S_{i,j} := U_S(\theta_i,\theta_j) = \sum_{n=1}^N \tilde{w}_{\theta_i,n}\, P_n(\theta_i,\theta_j)$. Then $U \in \mathbb{R}^{K \times K}$ is the payoff matrix, where row $i$ encodes the preferences of type $\theta_i$ and column $j$ corresponds to the generative distribution $\mathbb{P}_{\theta_j}$.

\begin{definition}[Direct revelation mechanism]
A (direct) mechanism is a mapping $\mu : \Theta \to \Theta$ assigning a reported control to each type.
\end{definition}

\begin{definition}[Incentive-Compatible Strategy]
A strategy $\mu$ is incentive compatible if $\mu(\theta_i)=\theta_i$ for all $i \in \{1,\dots,K\}$.
\end{definition}

\begin{definition}[Incentive compatibility]
A strategy $\mu$ is incentive compatible if
\begin{equation}
U_S(\theta_i,\mu(\theta_i))
\ge
U_S(\theta_i,y),
\quad \forall i \in \{1,\dots,K\},\; y \in \Theta.
\label{eq:IC_general}
\end{equation}
\end{definition}

\begin{proposition}[Matrix characterization of incentive compatibility]
\label{thm:IC_matrix}
The strategy $\mu(\theta_i)=\theta_i$ is incentive compatible if and only if
\begin{equation}
U^S_{i,i} \ge U^S_{i,j},
\quad \forall i,j \in \{1,\dots,K\}.
\label{eq:IC_discrete_full}
\end{equation}
\end{proposition}

\begin{proposition}[Row-wise dominance condition]
\label{thm:dominance}
Incentive-Compatible separation holds if and only if $U^S_{i,i} = \max_{j \in \{1,\dots,K\}} U^S_{i,j}$ for all $i$. Equivalently, for each type $\theta_i$, the diagonal entry is a maximizer of the $i$-th row, i.e., $\theta_i \in \arg\max_j U^S_{i,j}$.
\end{proposition}

 \begin{proposition}[PBNE under row-wise dominance]
\label{prop:pbne_dominance}

Consider the direct revelation model with 
$\Theta = \mathcal{Y} = \{\theta_1,\dots,\theta_K\}$ and the payoff matrix 
of the sender $U \in \mathbb{R}^{K \times K}$ defined by 
$U^S_{i,j} = U_S(\theta_i, y=\theta_j)$. Suppose that the row-wise dominance condition holds, i.e. 
$U^S_{i,i} \ge U^S_{i,j}$ for all $i,j$. Then there exists a Perfect Bayesian Nash Equilibrium (PBNE)
$(\mu^*, \{\delta_\alpha^*\}_{\alpha \in \mathcal{A}}, \rho^*)$
with the following structure:

\begin{enumerate}
\item[(i)] \textbf{Sender strategy:} 
The sender plays the incentive-compatible strategy $\mu^*(\theta_i)=\theta_i$ for all $i$.

\item[(ii)] \textbf{Beliefs:} 
For every $\alpha \in \mathcal{A}$ and observed score $s \in \mathbb{R}$, beliefs are updated via Bayes’ rule as
\[
\rho^*(\theta \mid s, \alpha)
=
\frac{
\rho(\theta)\, f_{\alpha,\theta}(s \mid \mu^*(\theta))
}{
\sum_{\theta' \in \Theta}
\rho(\theta')\, f_{\alpha,\theta'}(s \mid \mu^*(\theta'))
}.
\]
Under the Gaussian approximation, we have 
$S_\alpha \mid (\theta, \mu^*(\theta)) \approx \mathcal{N}(\bar{s}_{\alpha,\theta}(\mu^*(\theta)), \sigma^2_{\alpha,\theta}(\mu^*(\theta)))$, 
so $f_{\alpha,\theta}(\cdot \mid y)$ is given by \eqref{eq:likelihood}.

\item[(iii)] \textbf{Receiver strategy:} 
For each $\alpha \in \mathcal{A}$, the receiver applies the threshold rule
\[
\delta_\alpha^*(s)
=
\begin{cases}
\text{allow}, & s \le \eta(\alpha),\\
\text{block}, & s > \eta(\alpha),
\end{cases}
\]
which is optimal given the posterior $\rho^*(\cdot \mid s,\alpha)$.
\end{enumerate}

\noindent
Moreover, the equilibrium is separating in distribution: each type $\theta_i$ induces a distinct score distribution $f_{\alpha,\theta_i}(\cdot \mid \theta_i)$, so types are statistically identifiable though not perfectly revealed pointwise.

\end{proposition}

\begin{proof}
With row-wise dominance, for each type $\theta_i$, we have 
$U_S(\theta_i,\theta_i) \ge U_S(\theta_i,\theta_j)$ for all $j$, so the incentive-compatible strategy $\mu^*(\theta_i)=\theta_i$ is optimal. 
Given $\mu^*$, each type $\theta$ induces a score distribution 
$S_\alpha \mid (\theta, \mu^*(\theta)) \sim f_{\alpha,\theta}(\cdot \mid \mu^*(\theta))$, 
which is approximately Gaussian. Hence, beliefs $\rho^*(\theta \mid s,\alpha)$ follow the Bayes rule, establishing belief consistency.

Given these beliefs, the receiver solves a binary decision problem. By likelihood-ratio optimality, the optimal rule is a threshold rule in $s$, yielding $\delta_\alpha^*$. 
Off-path beliefs do not affect incentives since incentive-compatible reporting is optimal under row-wise dominance. 
Therefore, $(\mu^*, \{\delta_\alpha^*\}, \rho^*)$ satisfies sequential rationality and belief consistency and thus constitutes a PBNE.
\end{proof}

\subsection{Awareness and Statistical Discrimination}

\begin{definition}[Ordering of awareness types]
\label{def:awareness_order}

For $\alpha, \alpha' \in \mathcal{A}$, we say that $\alpha$ is
\emph{more aware} than $\alpha'$ (denoted $\alpha \succeq \alpha'$)
if, for all $y \in \mathcal{Y}$,
\[
\bar{s}_{\alpha,\theta_{\mathrm{mal}}}(y)
-
\bar{s}_{\alpha,\theta_{\mathrm{legit}}}(y)
\;\ge\;
\bar{s}_{\alpha',\theta_{\mathrm{mal}}}(y)
-
\bar{s}_{\alpha',\theta_{\mathrm{legit}}}(y).
\]

If the inequality is strict for some $y$, we write $\alpha \succ \alpha'$.
\end{definition}

\begin{proposition}[Awareness, threshold monotonicity, and behavior]
\label{prop:awareness_order}

Consider the Gaussian score model
$S_\alpha \mid (\theta,y) \sim \mathcal{N}(\bar{s}_{\alpha,\theta}(y), \sigma^2_\alpha(y))$
and suppose $A = \{\text{allow}, \text{block}\}$ with decision rule
$\delta_\alpha^*(s) = \mathbf{1}\{s > \eta(\alpha)\}$. 
Let the awareness ordering $\alpha \succeq \alpha'$ be defined as in
Definition~\ref{def:awareness_order}, i.e., for all $y$,
\[
\bar{s}_{\alpha,\theta_{\mathrm{mal}}}(y)
-
\bar{s}_{\alpha,\theta_{\mathrm{legit}}}(y)
\;\ge\;
\bar{s}_{\alpha',\theta_{\mathrm{mal}}}(y)
-
\bar{s}_{\alpha',\theta_{\mathrm{legit}}}(y).
\]

Then:

\begin{enumerate}
\item[(a)] \textbf{Likelihood-ratio monotonicity:}
For each $\alpha$, the likelihood ratio $\Lambda_\alpha(s)$ is increasing in $s$. Moreover, if $\alpha \succeq \alpha'$, then the signal-to-noise ratio
$\frac{\bar{s}_{\alpha,\theta_{\mathrm{mal}}}(y)-\bar{s}_{\alpha,\theta_{\mathrm{legit}}}(y)}{\sigma_\alpha(y)}$
is (weakly) larger, implying stronger statistical discrimination.

\item[(b)] \textbf{Threshold monotonicity:}
If $\alpha \succeq \alpha'$, then $\eta(\alpha) \le \eta(\alpha')$.
Thus, more aware receivers apply (weakly) stricter thresholds.

\item[(c)] \textbf{Acceptance ordering:}
For any $(\theta,y)$, letting
$P_\alpha(\theta,y) = \mathbb{P}(S_\alpha \le \eta(\alpha)\mid \theta,y)$,
we have
\[
P_\alpha(\theta,y) \le P_{\alpha'}(\theta,y)
\quad \text{whenever } \alpha \succeq \alpha'.
\]

\item[(d)] \textbf{Behavioral monotonicity:}
For any $s \in \mathbb{R}$,
\[
\delta_\alpha^*(s) \ge \delta_{\alpha'}^*(s)
\quad \text{whenever } \alpha \succeq \alpha',
\]
i.e., more aware receivers are (weakly) more likely to block.
\end{enumerate}

\end{proposition}

\begin{proof}
Under the Gaussian model with variance independent of $\theta$, the log-likelihood ratio is affine in $s$:
\[
\log \Lambda_\alpha(s)
=
\frac{
\bar{s}_{\alpha,\theta_{\mathrm{mal}}}(y)
-
\bar{s}_{\alpha,\theta_{\mathrm{legit}}}(y)
}{
\sigma_\alpha^2(y)
}\, s
+ \text{const}.
\]
Hence $\Lambda_\alpha(s)$ is increasing in $s$, and its slope is proportional to the separation of means, establishing (a). 
By optimality, the receiver applies a likelihood-ratio test, which reduces to a threshold rule in $s$. A larger separation between $\bar{s}_{\alpha,\theta_{\mathrm{mal}}}(y)$ and $\bar{s}_{\alpha,\theta_{\mathrm{legit}}}(y)$ implies that smaller values of $s$ suffice to detect the malicious type. Hence $\eta(\alpha) \le \eta(\alpha')$ when $\alpha \succeq \alpha'$, proving (b). 
Given $\eta(\alpha) \le \eta(\alpha')$, we have
$\{s \le \eta(\alpha)\} \subseteq \{s \le \eta(\alpha')\}$, which implies
$P_\alpha(\theta,y) \le P_{\alpha'}(\theta,y)$, establishing (c). 
Finally, since $\delta_\alpha^*(s) = \mathbf{1}\{s > \eta(\alpha)\}$ and $\eta(\alpha) \le \eta(\alpha')$, it follows that
$\delta_\alpha^*(s) \ge \delta_{\alpha'}^*(s)$ for all $s$, proving (d).
\end{proof}

\begin{proposition}[Keyword enrichment and strict awareness ordering]
\label{prop:keyword_awareness_strong}

Suppose the scoring function is given by
$\psi_\alpha(w) = \mathbf{1}\{w \in K_\alpha\}$ for a family of keyword sets
$\{K_\alpha\}_{\alpha \in \mathcal{A}}$, and let
\[
\bar{s}_{\alpha,\theta}(y)
=
\sum_{k=1}^L \mathbb{P}(W_k \in K_\alpha \mid \theta,y).
\]
Let $\Delta K := K_\alpha \setminus K_{\alpha'}$ denote the additional
keywords of $\alpha$ relative to $\alpha'$. 
Assume that for all $y \in \mathcal{Y}$,
\[
\mathbb{P}(W_k \in \Delta K \mid \theta_{\mathrm{mal}},y)
\;\ge\;
\mathbb{P}(W_k \in \Delta K \mid \theta_{\mathrm{legit}},y).
\]
Then:
\begin{enumerate}
\item[(i)] \textbf{Additive discrimination decomposition:}
\[
\bar{s}_{\alpha,\theta}(y)
=
\bar{s}_{\alpha',\theta}(y)
+
\sum_{k=1}^L
\mathbb{P}(W_k \in \Delta K \mid \theta,y).
\]

\item[(ii)] \textbf{Weak awareness ordering:}
\[
\bar{s}_{\alpha,\theta_{\mathrm{mal}}}(y)
-
\bar{s}_{\alpha,\theta_{\mathrm{legit}}}(y)
\;\ge\;
\bar{s}_{\alpha',\theta_{\mathrm{mal}}}(y)
-
\bar{s}_{\alpha',\theta_{\mathrm{legit}}}(y),
\]
and hence $\alpha \succeq \alpha'$.

\item[(iii)] \textbf{Strict awareness (if informative expansion):}
If the inequality is strict for some $y$, then $\alpha \succ \alpha'$.

\end{enumerate}

Consequently, awareness increases if and only if the additional
features $\Delta K$ contribute positively to type discrimination.
In particular, keyword set inclusion $K_\alpha \supseteq K_{\alpha'}$
improves awareness if and only if the newly added keywords are
informative, i.e., more likely under $\theta_{\mathrm{mal}}$ than
$\theta_{\mathrm{legit}}$.

\end{proposition}

\subsection{Incentive Design for Universal Benign Semantics}

We study the incentive structure under which the equilibrium induces
\emph{uniform benign behavior}, i.e.,
\(
\mu(\theta) = y_0, \ \forall \theta \in \Theta,
\)
so that all types, legitimate or malicious, select the same baseline
semantic control $y_0 \in \mathcal{Y}$. This corresponds to a pooling
outcome in which no sender has an incentive to deviate to alternative
semantic controls associated with phishing or manipulation.

From \eqref{eq:sender_payoff_gaussian}, the sender's expected payoff is
$U_S(\theta,y)$ for $\theta \in \Theta$ and $y \in \mathcal{Y}$, defining
a mapping $U_S : \Theta \times \mathcal{Y} \to \mathbb{R}$. For
finite sets $\Theta=\{\theta_1,\dots,\theta_K\}$ and
$\mathcal{Y}=\{y_0,y_1,\dots,y_M\}$, define
$U^S_{i,j} := U_S(\theta_i, y_j)$, yielding a $K \times (M+1)$ payoff
array; in particular, $\Theta$ and $\mathcal{Y}$ need not coincide.

The strategy $\mu(\theta)=y_0$ is optimal for all types if and only if it
solves \eqref{eq:sender_problem} for every $\theta \in \Theta$, which is
equivalent to the condition
\begin{equation}
U_S(\theta, y_0)
\;\ge\;
U_S(\theta, y),
\quad
\forall \theta \in \Theta,\; \forall y \in \mathcal{Y}.
\label{eq:benign_pooling_condition_general}
\end{equation}
In discrete notation, this becomes
\(
U^S_{i,0} \ge U^S_{i,j}, \ \forall i,j,
\)
i.e., the column corresponding to $y_0$ dominates all other columns
row-wise.

By \eqref{eq:sender_payoff_gaussian} and
\eqref{eq:acceptance_prob_sender}, the payoff $U_S(\theta,y)$ is driven
by the acceptance probability
\[
\mathbb{P}\!\left(S_{\alpha,L} \le \eta(\alpha)\mid \theta,y,\alpha\right),
\]
which depends on the Gaussian parameters
$\bar{s}_{\alpha,\theta}(y)$ and $\sigma_{\alpha,\theta}(y)$ induced by
the semantic control $y$. Hence, \eqref{eq:benign_pooling_condition_general}
requires that, for every type $\theta$, the control $y_0$ yields a
(weakly) higher expected payoff than any alternative $y$, thereby
penalizing deviations through their impact on the induced score
distribution and acceptance probabilities.

\begin{proposition}[Benign pooling equilibrium]
Let $U_S : \Theta \times \mathcal{Y} \to \mathbb{R}$ be the sender payoff
defined in \eqref{eq:sender_payoff_gaussian}. If the incentive condition
\eqref{eq:benign_pooling_condition_general} holds, then there exists a
Perfect Bayesian Nash Equilibrium
$(\mu^*, \{\delta_\alpha^*\}_{\alpha \in \mathcal{A}}, \rho^*)$ such that
\[
\mu^*(\theta) = y_0, \quad \forall \theta \in \Theta,
\]
where, for each $\alpha \in \mathcal{A}$, the receiver strategy
$\delta_\alpha^*$ is a best response to the induced pooling distribution
and is characterized by the threshold rule \eqref{eq:threshold_rule}, and
the belief system $\rho^*$ is consistent with Bayes' rule wherever
applicable.
\end{proposition}
 
 \subsubsection{Design Levers}

The condition \eqref{eq:benign_pooling_condition_general} can be enforced
by modifying the primitives entering \eqref{eq:sender_payoff_gaussian}.
These interventions operate through the statistical properties of the
signal, the cost structure faced by the sender, and the composition of
the receiver population.

\paragraph{Detection (statistical separation).}
Using \eqref{eq:acceptance_prob_sender}, one can reduce acceptance
probabilities under $y \neq y_0$ by increasing the statistical
separability between benign and malicious semantics. This can be
interpreted as improving detection mechanisms such as spam filters or
LLM-based classifiers, which increase the gap between
$\bar{s}_{\alpha,\theta_{\mathrm{mal}}}(y)$ and
$\bar{s}_{\alpha,\theta_{\mathrm{legit}}}(y)$, or tightening the
decision thresholds $\eta(\alpha)$. As a result, phishing messages are
more likely to be rejected, lowering the payoff $U_S(\theta,y)$ for
$y \neq y_0$.

\paragraph{Cost shaping (adversarial friction).}
An alternative approach is to increase $C_\theta(y)$ in
\eqref{eq:sender_value_structured} for $y \neq y_0$, thereby directly
penalizing malicious semantic controls. In practice, this corresponds to
introducing guardrails or constraints in the generative process, such as
content moderation, rate limiting, or filtering mechanisms that make it
more difficult to produce phishing messages. These mechanisms increase
the effort required for adversaries to evade detection, effectively
raising the cost of malicious behavior and reducing the corresponding
payoff.

\paragraph{Population shaping (awareness and training).}
Finally, one can adjust the weighting function $w(\alpha)$ or the
distribution $F$ in \eqref{eq:sender_payoff_gaussian} to place greater
emphasis on high-awareness receivers. This corresponds to improving user
awareness through training, education, or interface design, such as
warning systems or explainable alerts. A more informed population reduces
the acceptance probability of malicious messages across a wider range of
receivers, thereby reinforcing detection and further diminishing the
incentive to deviate from $y_0$.

\subsubsection{Optimal population shaping under discrete receivers.}

Assume that the receiver population is discrete, with types
$\{\alpha_1,\dots,\alpha_N\} \subset \mathcal{A}$. The population
distribution is given by
\[
F = \sum_{n=1}^N p_n \, \delta_{\alpha_n}, 
\quad
p_n \ge 0,\;\; \sum_{n=1}^N p_n = 1,
\]
where $p = (p_1,\dots,p_N) \in \Delta^N$. Let $F_0$ denote a baseline
distribution with weights $p^0 = (p_1^0,\dots,p_N^0)$. 
Under this representation, the sender payoff
\eqref{eq:sender_payoff_gaussian} becomes
\[
U_S(\theta,y)
=
\sum_{n=1}^N p_n \, w(\alpha_n)
\Big(v_\theta(\alpha_n) - C_\theta(y)\Big)
P_{\alpha_n}(\theta,y),
\]
where $P_{\alpha_n}(\theta,y)$ is defined in
\eqref{eq:acceptance_prob_sender}.

\begin{definition}[Pooling feasibility under population shaping]
Given a baseline distribution $p^0 \in \Delta^N$, a distribution
$p \in \Delta^N$ is said to be \emph{pooling-feasible} if it satisfies
\begin{equation}
\sum_{n=1}^N p_n \, \Delta_n(\theta,y) \ge 0,
\quad
\forall \theta \in \Theta,\;\forall y \in \mathcal{Y},
\label{eq:pooling_constraint_discrete}
\end{equation}
where
\[
\Delta_n(\theta,y)
:=
w(\alpha_n)
\Big[
\big(v_\theta(\alpha_n) - C_\theta(y_0)\big)
P_{\alpha_n}(\theta,y_0)
-
\big(v_\theta(\alpha_n) - C_\theta(y)\big)
P_{\alpha_n}(\theta,y)
\Big].
\]
\end{definition}

The quantity $\Delta_n(\theta,y)$ represents the marginal incentive
advantage of the benign control $y_0$ over a deviation $y$ under receiver
type $\alpha_n$. When $\Delta_n(\theta,y) > 0$, the deviation is
disincentivized, whereas $\Delta_n(\theta,y) < 0$ indicates that it
remains profitable. This yields a natural interpretation in terms of receiver awareness.
High-awareness receivers induce low acceptance probabilities for
malicious messages, leading to $\Delta_n(\theta,y) > 0$, while
low-awareness receivers may satisfy $\Delta_n(\theta,y) < 0$ and thus
create incentives for deviation.

\begin{problem}[Optimal population shaping]
Given a baseline distribution $p^0 \in \Delta^N$, the platform seeks to
minimally perturb the population so as to enforce the benign pooling
condition. This leads to the constrained optimization problem
\begin{equation}
\begin{aligned}
\min_{p \in \Delta^N} \quad & D(p \,\|\, p^0) \\
\text{s.t.} \quad 
& \sum_{n=1}^N p_n \, \Delta_n(\theta,y) \ge 0,
\quad \forall \theta \in \Theta,\; \forall y \in \mathcal{Y},
\end{aligned}
\label{eq:population_shaping_problem}
\end{equation}
where $\Delta_n(\theta,y)$ is defined in
\eqref{eq:pooling_constraint_discrete}.

A particularly important choice is the Kullback--Leibler divergence
\[
D(p\|p^0) = \sum_{n=1}^N p_n \log \frac{p_n}{p_n^0},
\]
which yields the \emph{minimum-information} or \emph{minimum-distortion}
adjustment of the population distribution required to eliminate
profitable deviations. In this sense, the platform implements the
smallest intervention (in an information-theoretic sense) that enforces
incentive compatibility.
\end{problem}

 \begin{proposition}[Optimal population shaping under KL divergence]
Consider problem \eqref{eq:population_shaping_problem} with
\[
D(p\|p^0) = \sum_{n=1}^N p_n \log \frac{p_n}{p_n^0},
\]
where $p^0 \in \Delta^N$ satisfies $p_n^0 > 0$ for all $n$. Suppose that
the feasible set of \eqref{eq:population_shaping_problem} is nonempty.
Then:

\begin{enumerate}
\item[(i)] There exists a unique optimal solution $p^* \in \Delta^N$.

\item[(ii)] There exist nonnegative multipliers
$\{\lambda_{\theta,y}^* \ge 0 : \theta \in \Theta, y \in \mathcal{Y}\}$,
associated with the pooling constraints
\eqref{eq:pooling_constraint_discrete}, such that $p^*$ admits the
exponential tilting form
\begin{equation}
p_n^*
=
\frac{
p_n^0 \exp\!\big( \Lambda_n^* \big)
}{
\sum_{m=1}^N p_m^0 \exp\!\big( \Lambda_m^* \big)
},
\label{eq:exp_tilting_solution_kl_clean_refined}
\end{equation}
where
\[
\Lambda_n^*
=
\sum_{\theta \in \Theta}\sum_{y \in \mathcal{Y}}
\lambda_{\theta,y}^* \, \Delta_n(\theta,y).
\]

\item[(iii)] (\textbf{Semi-closed form dual characterization})
The multiplier vector $\lambda^* = (\lambda_{\theta,y}^*)$ is the unique solution
of the nonlinear system
\begin{equation}
\sum_{n=1}^N
p_n^0 \exp\!\big( \Lambda_n^* \big)\, \Delta_n(\theta,y)
=
0,
\quad \forall (\theta,y) \text{ such that } \lambda_{\theta,y}^* > 0,
\label{eq:dual_moment_condition}
\end{equation}
and equivalently satisfies the normalized moment condition
\begin{equation}
\sum_{n=1}^N p_n^* \Delta_n(\theta,y) = 0
\quad \text{for all binding constraints.}
\label{eq:moment_matching_condition}
\end{equation}

\item[(iv)] The multipliers $\lambda_{\theta,y}^*$ satisfy the complementary
slackness conditions
\[
\lambda_{\theta,y}^*
\Big(
\sum_{n=1}^N p_n^* \Delta_n(\theta,y)
\Big)
= 0,
\quad \forall \theta \in \Theta,\; y \in \mathcal{Y}.
\]

\item[(v)] The induced distribution
\(
F^* = \sum_{n=1}^N p_n^* \delta_{\alpha_n}
\)
satisfies the benign pooling condition
\eqref{eq:benign_pooling_condition_general}. Consequently, there exists a
Perfect Bayesian Nash Equilibrium
$(\mu^*, \{\delta_\alpha^*\}_{\alpha \in \mathcal{A}}, \rho^*)$
such that
\[
\mu^*(\theta) = y_0, \quad \forall \theta \in \Theta.
\]
\end{enumerate}
\end{proposition}

\begin{proof}
Since $p_n^0 > 0$ for all $n$, the Kullback--Leibler divergence
$D(p\|p^0)$ is strictly convex on $\Delta^N$. The feasible set is a
nonempty, closed, and convex subset of the compact simplex $\Delta^N$,
hence a unique optimizer $p^*$ exists. 
Introduce Lagrange multipliers $\lambda_{\theta,y} \ge 0$ for each
constraint
\(
\sum_{n} p_n \Delta_n(\theta,y) \ge 0
\)
and $\nu$ for the normalization constraint. The Lagrangian is
\[
\mathcal{L}(p,\lambda,\nu)
=
\sum_{n=1}^N p_n \log \frac{p_n}{p_n^0}
-
\sum_{\theta,y} \lambda_{\theta,y}
\Big( \sum_{n=1}^N p_n \Delta_n(\theta,y) \Big)
+
\nu \Big( \sum_{n=1}^N p_n - 1 \Big).
\]

The first-order optimality condition yields
\[
\log \frac{p_n^*}{p_n^0} + 1
=
\sum_{\theta,y} \lambda_{\theta,y}^* \Delta_n(\theta,y)
+ \nu^*,
\quad \forall n.
\]

Exponentiating and normalizing over $n$ gives the exponential tilting form
\eqref{eq:exp_tilting_solution_kl_clean_refined}, where the normalization
constant is
\[
Z(\lambda^*) = \sum_{m=1}^N p_m^0 \exp\!\big( \Lambda_m^* \big),
\quad
\nu^* = 1 - \log Z(\lambda^*).
\]

Substituting \eqref{eq:exp_tilting_solution_kl_clean_refined} into the
constraint functions yields
\[
\sum_{n=1}^N p_n^* \Delta_n(\theta,y)
=
\frac{
\sum_{n=1}^N p_n^0 \exp\!\big( \Lambda_n^* \big) \Delta_n(\theta,y)
}{
Z(\lambda^*)
}.
\]

Hence, for each binding constraint, the numerator must vanish, which gives
the system \eqref{eq:dual_moment_condition}. Dividing by $Z(\lambda^*)$
yields the equivalent normalized moment condition
\eqref{eq:moment_matching_condition}. This establishes the semi-closed
form characterization of $\lambda^*$. 
Complementary slackness follows from the KKT conditions. The equilibrium
implication follows from the feasibility of $p^*$ and the satisfaction of
the benign pooling constraints.
\end{proof}

Under the optimal solution \eqref{eq:exp_tilting_solution_kl_clean_refined}, the
population is reweighted as $p_n^* \propto p_n^0 \exp(\Lambda_n^*)$, where
$\Lambda_n^* = \sum_{\theta,y} \lambda_{\theta,y}^* \Delta_n(\theta,y)$.
Thus, receiver types that more strongly discourage deviations (i.e.,
with larger $\Delta_n(\theta,y)$) are amplified, while those that enable
deviations are down-weighted, shifting the effective population toward
more robust receivers.

The quantity $\Delta_n(\theta,y)$ can be improved through defensive
mechanisms. Training increases user awareness, reducing acceptance of
malicious messages and raising $\Delta_n(\theta,y)$. Interface design,
such as warnings or decision aids, lowers the acceptance probability
$P_{\alpha_n}(\theta,y)$ and thereby increases $\Delta_n(\theta,y)$.
Segmentation or filtering reduces the influence of receiver types with
$\Delta_n(\theta,y) < 0$ by limiting their exposure.

\section{Extensions}

\subsection{Two-sided scoring model}
We assume that, from the receiver's perspective, the null hypothesis
$H_0$ and the associated scoring rule are fixed across messages. In the
baseline model, the receiver evaluates a message using a keyword-based
scoring function that identifies suspicious content.
In practice, receivers may maintain both (i) a list of suspicious
keywords that trigger alerts and (ii) a list of indicative keywords that
signal legitimacy. In the current formulation, we focus only on the
former: keywords that raise suspicion.

The receiver generates an alert when the message contains a token from a
predefined suspicious keyword set. Thus detection is driven by
alert-triggering terms. An alternative behavioral model is that some
receivers may also interpret the \emph{absence} of expected legitimate
keywords as evidence of deception. In such cases, both the presence of
suspicious terms and the absence of credibility-inducing terms contribute
to the detection decision.

\subsection{Systematic Blindness and Information Loss}

We formalize the notion of \emph{systematic blindness} as a structural
limitation in the receiver’s perceptual representation, and characterize
its implications for information loss and learning dynamics.

\paragraph{Informative features.}
Let $\theta \in \{\text{legit}, \text{mal}\}$ denote the sender type and
define the set of informative features
\[
\mathcal{I}
=
\left\{
w \in V :
P(w \mid \theta = \text{mal})
>
P(w \mid \theta = \text{legit})
\right\}.
\]
These are the features that carry statistical evidence of malicious intent.

\begin{definition}[Systematic Blindness]
A receiver $\alpha$ is said to exhibit \emph{systematic blindness} if
\(
\mathcal{I} \setminus K_\alpha \neq \emptyset.
\)
Equivalently, there exists a nonempty set of informative features
\(
\Delta K_\alpha := \mathcal{I} \setminus K_\alpha
\) 
that are never observed by the receiver and therefore do not contribute
to the score $S_\alpha(m)$.
\end{definition}

Systematic blindness is a structural property: the receiver does not
misinterpret certain features; rather, these features are entirely absent
from the receiver’s representation. Let $V$ denote the token space and
$K_\alpha \subseteq V$ the receiver’s accessible feature set.

This induces a projection operator at the token level,
\(
\pi_\alpha : V \to K_\alpha \cup \{\varnothing\},
\)
which filters out features not contained in $K_\alpha$. This projection
extends componentwise to messages of length $L$, yielding
\(
\Pi_\alpha : V^L \to K_\alpha^L, \ 
\Pi_\alpha(m) = (\pi_\alpha(w_1), \dots, \pi_\alpha(w_L)).
\)

The receiver therefore evaluates the score $S_\alpha(m)$ based on the
projected message $\Pi_\alpha(m)$ rather than the original message $m$. 
Systematic blindness thus induces a loss of statistical information about
the sender type, as components of the message that may be informative for
discrimination are removed prior to evaluation.

\begin{definition}[Effective Information Set]
The effective information available to receiver $\alpha$ is
\(
\mathcal{I}_\alpha := \mathcal{I} \cap K_\alpha.
\)
\end{definition}

\begin{proposition}[Information Loss under Blindness]
If $\alpha$ is systematically blind, then
\(
\mathcal{I}_\alpha \subsetneq \mathcal{I},
\)
and the resulting score distribution satisfies
\[
\bar{s}_\alpha(y)
=
\sum_{k=1}^L P(W_k \in K_\alpha \mid y)
<
\sum_{k=1}^L P(W_k \in \mathcal{I} \mid y).
\]
\end{proposition}

\begin{proof}
Since $\mathcal{I}_\alpha = \mathcal{I} \cap K_\alpha$ and
$\mathcal{I} \setminus K_\alpha \neq \emptyset$, we have
$\mathcal{I}_\alpha \subsetneq \mathcal{I}$. The score aggregates only
features in $K_\alpha$, hence excludes contributions from
$\mathcal{I} \setminus K_\alpha$, which are informative but unobserved.
\end{proof}

Information loss manifests as a degradation in statistical
discriminability. In particular, the mean separation between malicious
and legitimate signals decreases,
\(
\mu_{\alpha,\text{mal}} - \mu_{\alpha,\text{legit}} \downarrow,
\)
which directly weakens the distinguishability of the two distributions.
At the same time, the signal-to-noise ratio in the Gaussian approximation
\(
S_\alpha \approx \mathcal{N}(\mu_\alpha, \sigma_\alpha^2)
\) 
is reduced, as informative components of the signal are systematically
excluded. As a consequence, the overlap between the score distributions
increases, leading to a higher acceptance probability for malicious
messages under the receiver’s decision rule.

 \begin{remark}[Vocabulary-Induced Blindness]
Systematic blindness can be interpreted as an epistemic limitation induced
by the receiver's vocabulary. A receiver of type $\alpha$ processes only
tokens in its accessible feature set $K_\alpha \subset V$. Consequently,
the effective information set is
$\mathcal{I}_\alpha=\mathcal{I}\cap K_\alpha$, and any informative feature
in $\mathcal{I}\setminus K_\alpha$ is not merely misinterpreted but
entirely unrepresented in the receiver's perceptual model.

This phenomenon reflects a structural constraint on knowledge: the
receiver's ability to infer the sender's type is bounded by the expressive
capacity of its feature space. In this sense, the vocabulary $K_\alpha$
defines the observable world of the receiver. Signals outside this space
are epistemically inaccessible, regardless of their statistical
informativeness. This corresponds to a form of representational
incompleteness in statistical decision-making, where sufficient statistics
are restricted by the class of models \cite{cover1999elements}.

This perspective is closely aligned with Wittgenstein's insight that ``the
limits of my language mean the limits of my world''
\cite{wittgenstein1922tractatus}. Here, language, formalized as the set of
features $K_\alpha$, acts as a projection operator that determines which
aspects of the message are observable and meaningful. In this sense, the
receiver operates under a constrained $\sigma$-algebra of observables,
inducing an endogenous information structure.

Systematic blindness arises when the true informative structure of the
environment exceeds this representational boundary. This phenomenon is
closely related to bounded rationality \cite{simon1957models}, where
decision-makers operate under limited cognitive and informational
resources, and to model misspecification in statistical inference
\cite{white1982maximum}, where the true data-generating process lies
outside the assumed model class.

Expanding $K_\alpha$ through learning, training, or representation enrichment enlarges the
effective observable space, thus refining the induced information structure and reducing
epistemic blindness. In this sense, learning can be interpreted as an expansion of the
receiver’s measurable feature space, analogous to refining partitions in information theory
or increasing the expressiveness of hypothesis classes in statistical learning.
\end{remark}

 \subsection{Mindset Dynamics}\label{sec:mindset-dynamics}
We model the evolution of the receiver’s mindset as a feature expansion
process. At time $t$, the receiver $\alpha$ is characterized by a 
set of characteristics $K_\alpha(t) \subseteq V$, which determines the components of the
message that are perceptible and contribute to the score
\[
S_\alpha(m;t) = \sum_{k=1}^L \mathbf{1}\{w_k \in K_\alpha(t)\}.
\]

\begin{definition}[Mindset Dynamics]
The feature set evolves according to
\[
K_\alpha(t+1)
=
K_\alpha(t)
\cup
\mathcal{L}_\alpha(t),
\]
where $\mathcal{L}_\alpha(t) \subseteq V$ denotes newly acquired features.
\end{definition}

The update set $\mathcal{L}_\alpha(t)$ captures the accumulation of
experience and information, arising from exposure, feedback, or external
signals. The update is cumulative, implying
\(
K_\alpha(t) \subseteq K_\alpha(t+1),
\)
so that the feature set grows monotonically over time. Consequently, the
effective information set
\(
\mathcal{I}_\alpha(t) = \mathcal{I} \cap K_\alpha(t)
\)
also expands, increasing the portion of informative features that the
receiver can exploit.

Under the Gaussian approximation
\(
S_\alpha(m;t) \approx \mathcal{N}(\mu_\alpha(t), \sigma_\alpha^2(t)),
\)
this expansion increases the mean separation between legitimate and
malicious messages, thereby improving detection performance. In
particular, if $\mathcal{L}_\alpha(t) \cap \mathcal{I} \neq \emptyset$
infinitely often, then the blind region shrinks over time, and vanishes
asymptotically if
\(
\lim_{t \to \infty} K_\alpha(t) \supseteq \mathcal{I}.
\)

This improvement is countered by strategic adaptation. Since the sender
controls the distribution $P(\cdot \mid y)$, it can shift probability
mass toward features outside the current $K_\alpha(t)$, preserving
evasion. The resulting interaction forms a co-evolutionary dynamic
\[
K_\alpha(t) \uparrow
\quad \Longleftrightarrow \quad
P(\cdot \mid y_t) \text{ shifts outward},
\]
in which learning expands the perceptual boundary while the sender
continuously relocates signals beyond it.

\section{Case Study and Numerical Experiments}\label{sec:case}
\newcommand{\CaseSamples}{10000}
\newcommand{\CaseLength}{96}
\newcommand{\CaseCLTMean}{34.41}
\newcommand{\CaseCLTVariance}{26.02}
\newcommand{\CaseCLTSkewness}{0.10}
\newcommand{\CaseCLTKurtosis}{-0.04}
\newcommand{\CaseCLTKS}{0.016}
\newcommand{\CaseGuardrailPenalty}{0.10}
\newcommand{\CaseCriticalNoPenalty}{0.48}
\newcommand{\CaseCriticalGuardrail}{0.18}
\newcommand{\CaseBaselineGap}{0.16}
\newcommand{\CaseBaselineGapGuardrail}{0.06}
\newcommand{\CaseMaxAcceptanceError}{0.010}
\newcommand{\CaseMeanAcceptanceError}{0.002}
\newcommand{\CaseMindsetInitialFixed}{0.572}
\newcommand{\CaseMindsetFinalFixed}{0.153}
\newcommand{\CaseMindsetPeakAdaptive}{0.960}
\newcommand{\CaseMindsetFinalAdaptive}{0.209}

We now present a reproducible numerical case study that illustrates the
main analytical results. The experiment is a controlled synthetic
LLM-channel model: each semantic control induces a distribution over
token-feature categories, and receiver awareness determines which
features are scored. This construction is not intended to benchmark a
particular deployed LLM. Its purpose is to isolate the mechanisms in
Sections~\ref{sec:llm-game}--\ref{sec:pbne}: Gaussian score
approximation, threshold-based detection, awareness-dependent
acceptance, mindset dynamics, and incentive design for benign pooling.
The numerical data and figures are generated by
\texttt{run\_case\_study\_experiments.py}.

\subsection{Experimental Setup}

The sender chooses one of three semantic controls,
\[
Y=\{y_0,y_1,y_2\}
=
\{\text{benign},\text{aggressive},\text{stealth}\}.
\]
The benign control represents routine administrative communication, the
aggressive control represents an explicit phishing attempt with strong
urgency and credential cues, and the stealth control represents a more
subtle phishing message that reduces explicit alert terms while
retaining enough action-oriented language to remain persuasive.

Each generated message has length \(L=\CaseLength\). At each token
position, the controlled channel draws one of eight bounded feature
categories: neutral text, legitimate-context cues, urgency, action
requests, credential requests, links or attachments, authority
impersonation, and contextual mismatch. These categories are a
low-dimensional proxy for the token-level features in
\eqref{eq:token_score}. The three controls differ only through their
feature probabilities, so the experiment directly implements the
controlled stochastic kernels in \eqref{eq:llm_chain_rule_formal}.
The vocabulary used to instantiate these categories is seeded by a
public phishing keyword list in the Hunting-Lists repository
\cite{cyb3rmik3_phishing_keywords_2023}. We group that list into
semantic cue families, including urgency, account status, password and
verification language, document-sharing language, invoice or payment
references, message notifications, and service-request language. The
list is used only to make the synthetic channel interpretable; it is not
treated as an empirical distribution over real phishing emails.

Figure~\ref{fig:case_feature_probabilities} reports the feature
probabilities used by the three controls. The benign control places most
mass on neutral and legitimate-context tokens. The aggressive control
places substantially more mass on urgency, action requests, credential
requests, and links. The stealth control is intermediate: it suppresses
the most explicit indicators relative to the aggressive control, but
keeps enough action-oriented, authority-framed, and context-mismatch
features to remain strategically useful.

\begin{figure}[!htbp]
\centering
\includegraphics[width=0.88\linewidth]{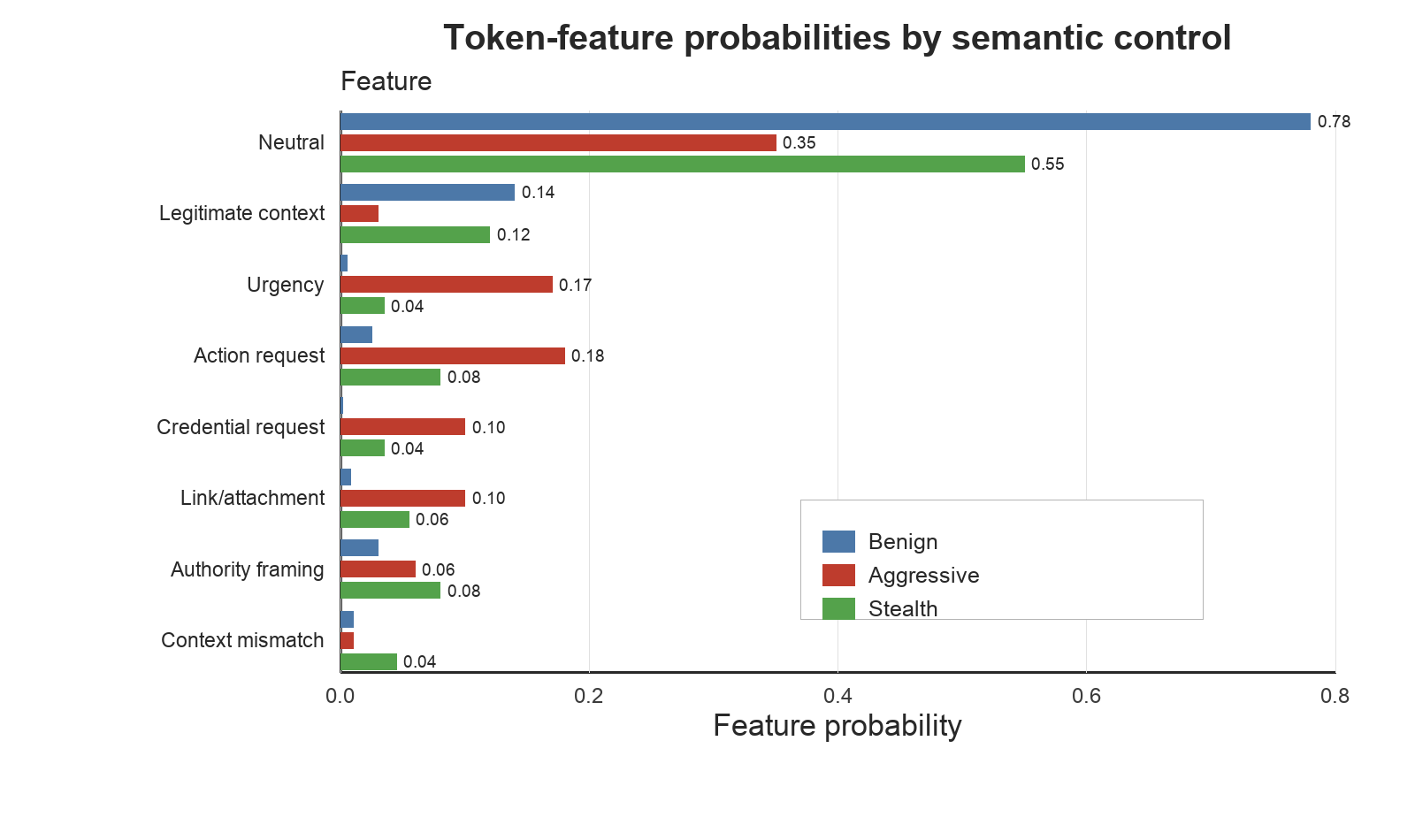}
\caption{Feature-category probabilities induced by the three semantic
controls. The stealth control reduces explicit phishing cues while
retaining weaker authority, action, and mismatch cues.}
\label{fig:case_feature_probabilities}\vspace{-2mm}
\end{figure}
\FloatBarrier

To make the controls concrete, the following quote blocks show
representative messages generated by instantiating the feature
categories with short phrase templates. These examples are not intended
to benchmark a deployed LLM; they illustrate how the abstract semantic
controls translate into messages with different observable cues.

\begin{quote}
\small
\textbf{Benign.} ``The operations team has posted the monthly reimbursement reminder and service request summary in the employee portal. Please review the notice when convenient and submit any corrections through the standard service desk form. No immediate action is required if your record is already accurate; this note is only intended to keep the directory and reimbursement workflow current.''

\emph{Salient cues:} routine administrative context; low urgency; standard internal workflow; no credential request; no account-loss threat.
\end{quote}
\begin{quote}
\small
\textbf{Aggressive.} ``Action required: your payroll account is scheduled for suspension today because the verification record is incomplete. Open the secure payroll validation page and confirm your password and MFA code before the close of business to prevent loss of access to direct deposit and benefits. Failure to complete this step may delay your next payment.''

\emph{Salient cues:} high urgency; explicit credential request; link-based action; threat of account suspension and payment disruption.
\end{quote}
\begin{quote}
\small
\textbf{Stealth.} ``Hi, I am reconciling the team directory before tomorrow's audit and noticed that a few contact records still need confirmation. When you have a moment, please review the internal request form and confirm whether the listed phone number, department, and backup contact are current. This will help the support team close the ticket before the audit packet is finalized.''

\emph{Salient cues:} mild urgency; audit-related authority framing; indirect action request; plausible business context; subtle mismatch in process.
\end{quote}

\FloatBarrier

Receivers have awareness types
\[
\mathcal{A}
=
\{\alpha_{\mathrm{naive}},
\alpha_{\mathrm{mid}},
\alpha_{\mathrm{aware}}\}.
\]
The scoring maps are nested. Naive receivers score only explicit cues
such as urgency, credential requests, and links. Intermediate receivers
also score action requests, authority framing, and weak contextual
signals. Aware receivers score the full feature set, including subtle
contextual mismatch. Thus higher awareness enriches the feature map
\(\Psi_\alpha\), as in Proposition~\ref{prop:keyword_awareness_strong}.

For each semantic control and awareness type, we generate
\(\CaseSamples\) messages and compute the cumulative score
\[
S_{\alpha,L}=\sum_{k=1}^{L}\psi_\alpha(W_k).
\]
The receiver uses the Gaussian threshold rule in
Proposition~\ref{prop:threshold_policy}. For each \(\alpha\), the
threshold is calibrated as
\[
\eta_\alpha
=
\frac{1}{2}
\left(
\widehat{\mu}_{\alpha,\mathrm{ben}}
+
\widehat{\mu}_{\alpha,\mathrm{mal}}
\right),
\]
where \(\widehat{\mu}_{\alpha,\mathrm{ben}}\) is the empirical mean
under \(y_0\), and \(\widehat{\mu}_{\alpha,\mathrm{mal}}\) is the
average empirical mean under \(y_1\) and \(y_2\). A message is accepted
when \(S_{\alpha,L}\le \eta_\alpha\) and blocked otherwise.

\subsection{Gaussian Score Approximation}

The first experiment tests Proposition~\ref{prop:score_clt}. We focus
on the stealth control under the aware receiver because this is the
hardest detection case: explicit phishing cues are muted, while subtle
features remain visible only to high-awareness receivers. Across
\(\CaseSamples\) simulated messages, the score has empirical mean
\(\CaseCLTMean\) and variance \(\CaseCLTVariance\). After
standardization, the skewness is \(\CaseCLTSkewness\), the excess
kurtosis is \(\CaseCLTKurtosis\), and the Kolmogorov distance from the
standard normal CDF is \(\CaseCLTKS\).

\begin{figure}[!htbp]
\centering
\includegraphics[width=0.80\linewidth]{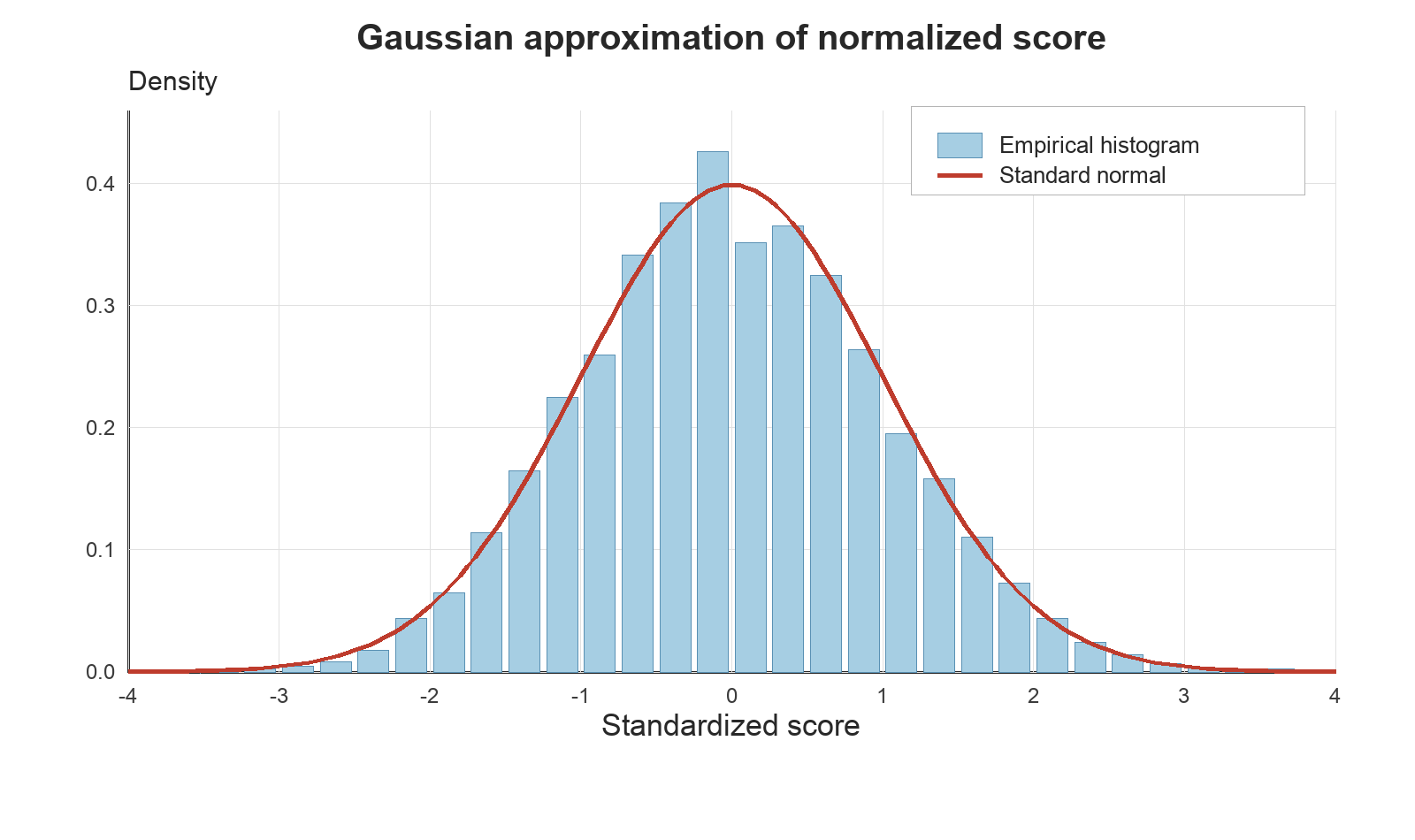}
\caption{Gaussian approximation for the normalized cumulative score
under the stealth semantic control and aware receiver. The close match
between the empirical histogram and the standard normal density
supports the CLT approximation in Proposition~\ref{prop:score_clt}.}
\label{fig:case_clt_validation}\vspace{-2mm}
\end{figure}

Figure~\ref{fig:case_clt_validation} shows that the standardized score
is close to Gaussian. This supports the use of the normal approximation
in \eqref{eq:score_distribution} and justifies the closed-form
acceptance probability in \eqref{eq:acceptance_probability}.

\subsection{Threshold Detection and Awareness Ordering}

Table~\ref{tab:case_scores} reports the empirical score means, standard
deviations, thresholds, and acceptance probabilities for each semantic
control and awareness type. The benign control is accepted with
probability one in the simulation, while the aggressive phishing control
is rejected for all awareness levels. The stealth control is the
strategically interesting case: it is accepted by naive receivers with
probability \(0.541\), by intermediate receivers with probability
\(0.297\), and by aware receivers with probability \(0.157\).

\begin{table}[!htbp]
\centering
\caption{Estimated score statistics and empirical acceptance probabilities.}
\label{tab:case_scores}
\scriptsize
\begin{tabular}{llrrrr}
\hline
Control & Awareness & Mean & SD & $\eta_\alpha$ & Accept \\
\hline
Benign & Naive & 3.33 & 1.52 & 18.64 & 1.000 \\
Benign & Intermediate & 5.97 & 2.16 & 24.74 & 1.000 \\
Benign & Aware & 7.85 & 2.76 & 29.25 & 1.000 \\
Aggressive & Naive & 49.54 & 4.87 & 18.64 & 0.000 \\
Aggressive & Intermediate & 59.91 & 5.03 & 24.74 & 0.000 \\
Aggressive & Aware & 66.91 & 5.47 & 29.25 & 0.000 \\
Stealth & Naive & 18.36 & 3.60 & 18.64 & 0.541 \\
Stealth & Intermediate & 27.10 & 4.24 & 24.74 & 0.297 \\
Stealth & Aware & 34.41 & 5.10 & 29.25 & 0.157 \\
\hline
\end{tabular}
\end{table}

The Gaussian approximation also predicts acceptance probabilities
accurately in this experiment. Across all control-awareness pairs, the
maximum absolute gap between empirical acceptance and the Gaussian
prediction is \(\CaseMaxAcceptanceError\), and the mean absolute gap is
\(\CaseMeanAcceptanceError\).

\begin{figure}[!htbp]
\centering
\includegraphics[width=0.82\linewidth]{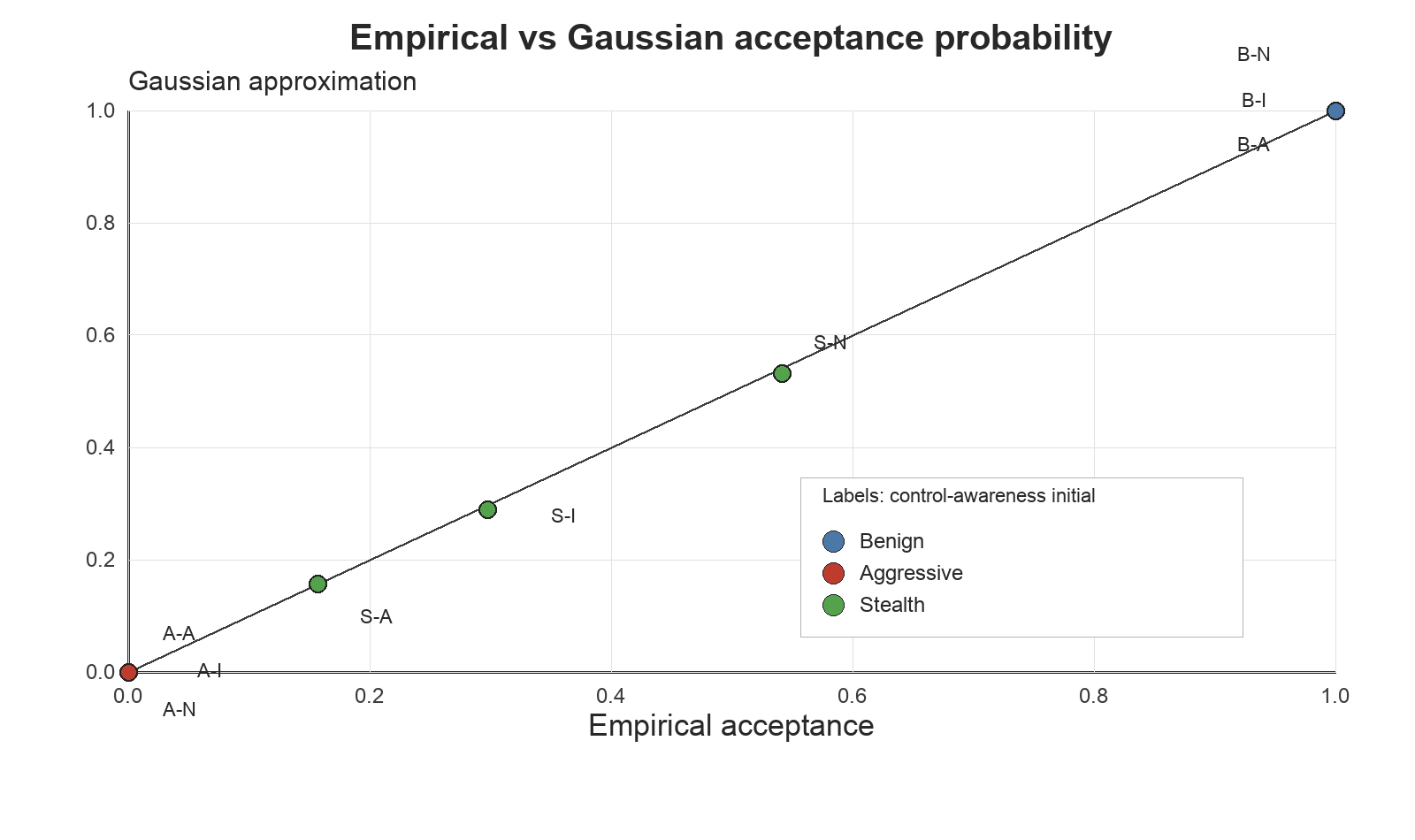}
\caption{Empirical acceptance probabilities versus Gaussian
approximations for all semantic-control and awareness-type pairs. Points
near the diagonal indicate that the closed-form approximation in
\eqref{eq:acceptance_probability} tracks the Monte Carlo estimates.}
\label{fig:case_acceptance_calibration}\vspace{-2mm}
\end{figure}

\begin{figure}[!htbp]
\centering
\includegraphics[width=0.80\linewidth]{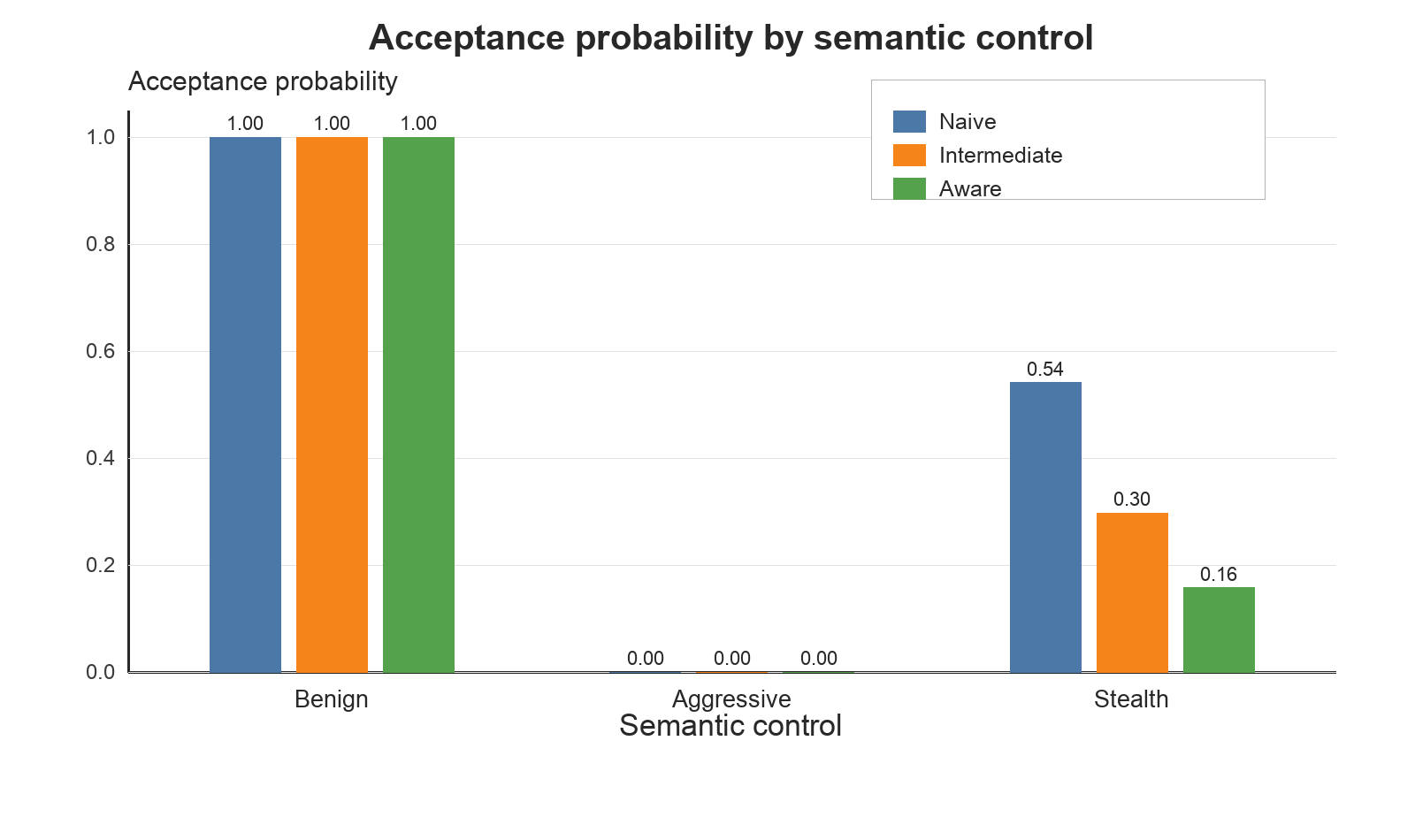}
\caption{Empirical acceptance probability across semantic controls and
awareness types. Stealth phishing selectively evades lower-awareness
receivers, while increased awareness reduces acceptance, illustrating
the behavioral ordering in Proposition~\ref{prop:awareness_order}.}
\label{fig:case_acceptance_awareness}\vspace{-2mm}
\end{figure}

Figure~\ref{fig:case_acceptance_awareness} visualizes the same effect.
The result illustrates how semantic control changes the score
distribution and how awareness reshapes the receiver's decision rule.
Aggressive phishing is easy to detect because it shifts the score far
above the threshold. Stealth phishing instead moves the score close to
the naive threshold while remaining more visible to richer feature maps.
Thus the experiment supports the acceptance-ordering part of
Proposition~\ref{prop:awareness_order}: higher awareness decreases the
acceptance probability of malicious semantic controls.

\FloatBarrier

\subsection{Mindset Dynamics and Adaptive Evasion}

We next instantiate the mindset dynamics model in
Section~\ref{sec:mindset-dynamics}. The receiver begins with an explicit
cue vocabulary consisting of urgency, credential requests, and links.
Across successive learning stages, the receiver adds action requests,
authority framing, and contextual mismatch to its feature set
\(K_\alpha(t)\). Thus the effective informative-feature coverage
\(|K_\alpha(t)\cap\mathcal{I}|/|\mathcal{I}|\) grows monotonically,
as in the update rule for \(\mathcal{L}_\alpha(t)\).

We compare two sender responses. In the fixed-stealth condition, the
sender keeps the stealth semantic control unchanged while the receiver
learns. In the adaptive-stealth condition, the sender reallocates
probability mass away from newly learned cues and toward cues that are
still outside the current feature set. This creates the co-evolutionary
pattern described in Section~\ref{sec:mindset-dynamics}: receiver
learning expands the perceptual boundary, while the sender attempts to
move persuasive evidence beyond it.

\begin{figure}[!htbp]
\centering
\includegraphics[width=0.82\linewidth]{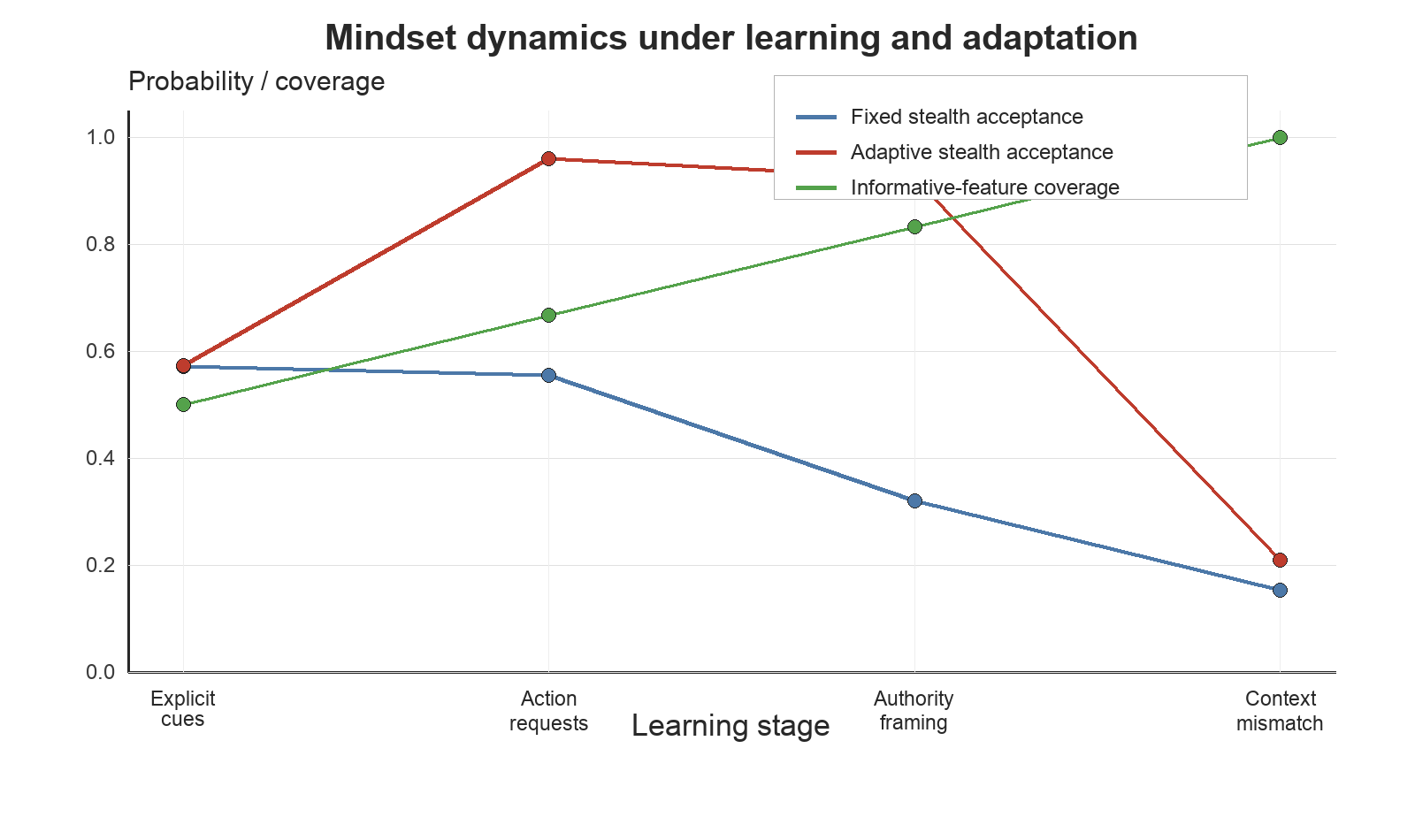}
\caption{Mindset dynamics under feature learning and sender adaptation.
Fixed stealth becomes easier to detect as \(K_\alpha(t)\) expands.
Adaptive stealth preserves evasion while informative cues remain outside
the receiver's feature set, but loses this advantage as coverage becomes
complete.}
\label{fig:case_mindset_dynamics}\vspace{-2mm}
\end{figure}

Figure~\ref{fig:case_mindset_dynamics} shows that learning alone reduces
the acceptance probability of the fixed stealth control from
\(\CaseMindsetInitialFixed\) at the initial feature set to
\(\CaseMindsetFinalFixed\) when all modeled informative categories are
represented. Adaptive stealth delays this improvement: its acceptance
probability reaches \(\CaseMindsetPeakAdaptive\) while blind regions
remain, but falls to \(\CaseMindsetFinalAdaptive\) once the receiver's
feature coverage is complete. This experiment supports the interpretation
of mindset dynamics as a moving-boundary problem rather than a static
threshold adjustment.

\FloatBarrier

\subsection{Mechanism Design and Benign Pooling}

The final experiment illustrates the benign-pooling incentive condition
from Section~\ref{sec:pbne}. We compute the malicious
sender's payoff from each semantic control using the Gaussian acceptance
model in \eqref{eq:sender_payoff}. Let \(q\in[0,1]\) denote the
share of aware receivers in the population. The remaining mass
\(1-q\) is split between naive and intermediate receivers in the same
proportion as the baseline population. For each \(q\), define the
deviation gain
\[
\Delta(q,\tau)
=
\max_{y\in\{y_1,y_2\}}U_S(\theta^{\mathrm{mal}},y;\tau)
-
U_S(\theta^{\mathrm{mal}},y_0;0),
\]
where \(\tau\) is an added cost applied to nonbenign controls. A
negative value of \(\Delta(q,\tau)\) means that the malicious type has
no profitable deviation from the benign control, so the row-wise
incentive condition for benign pooling is satisfied for that type.

\begin{figure}[!htbp]
\centering
\includegraphics[width=0.80\linewidth]{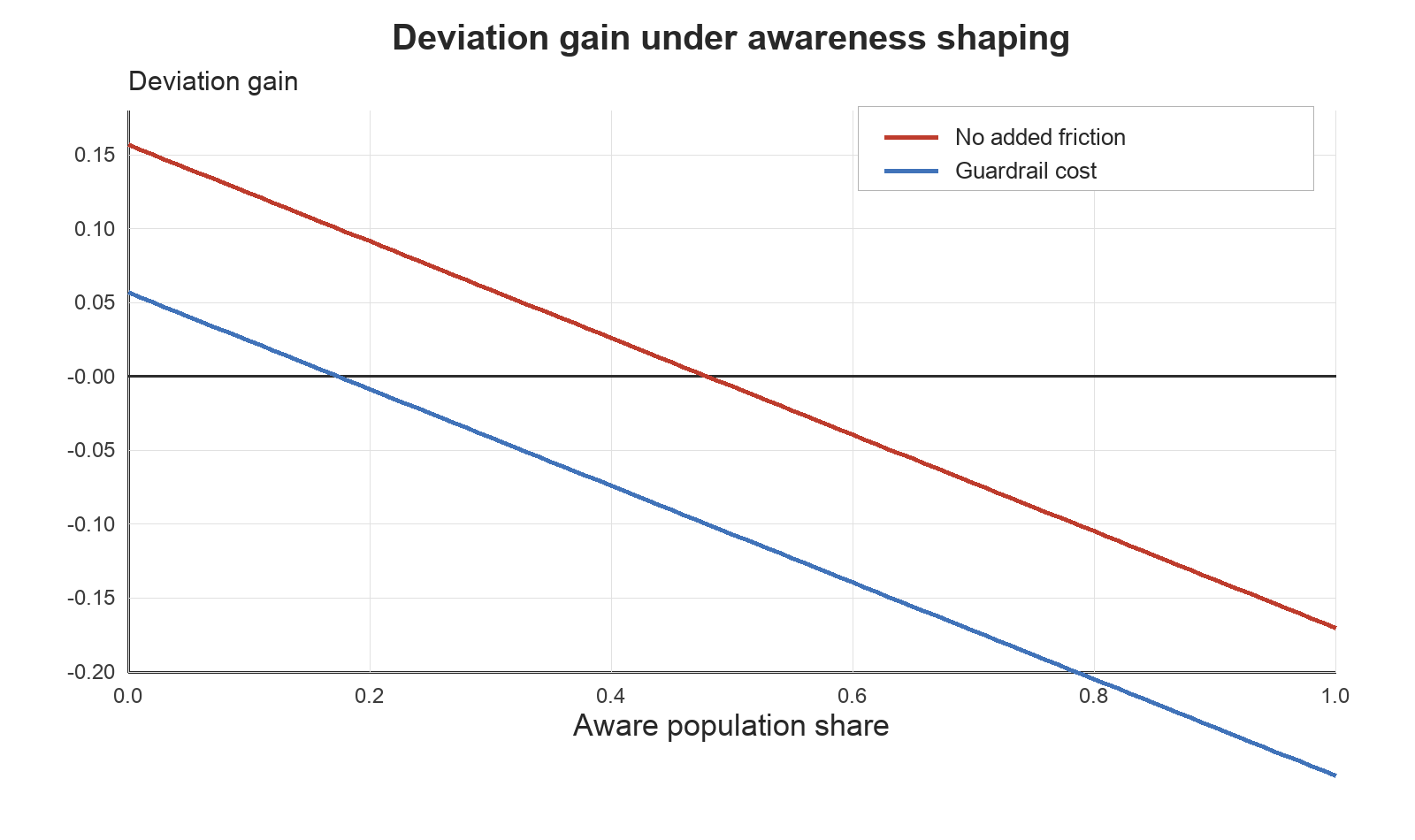}
\caption{Deviation gain for the malicious sender as the aware share of
the receiver population increases. The zero line marks the boundary of
the benign-pooling incentive condition. Awareness shaping and guardrail
costs both reduce profitable deviations.}
\label{fig:case_mechanism_design}\vspace{-2mm}
\end{figure}

Figure~\ref{fig:case_mechanism_design} shows that awareness shaping
alone can eliminate profitable malicious deviations once the aware share
reaches approximately \(\CaseCriticalNoPenalty\). With an added
guardrail cost of \(\tau=\CaseGuardrailPenalty\), the critical aware
share falls to approximately \(\CaseCriticalGuardrail\). At the baseline
population, the deviation gain is \(\CaseBaselineGap\) without the
guardrail and \(\CaseBaselineGapGuardrail\) with the guardrail. These
results illustrate the design levers in
Section~\ref{sec:pbne}: detection enrichment lowers malicious
acceptance probabilities, cost shaping penalizes deceptive controls, and
population shaping changes the effective receiver distribution in the
sender's payoff.

\FloatBarrier

\subsection{Discussion}

The numerical experiments support the paper's theoretical claims in a
single controlled setting. First, cumulative message scores behave
approximately normally, making the Gaussian detection formulas
operational. Second, semantic controls change acceptance probabilities
by shifting the induced score distribution. Third, receiver awareness is
not merely a lower threshold; it is a richer feature representation that
makes stealthy manipulation more visible. Fourth, mindset dynamics show
how learning reduces fixed stealth evasion while adaptive senders can
temporarily exploit blind regions. Finally, the mechanism-design
experiment shows how awareness shaping and adversarial friction can
convert a profitable deceptive deviation into a benign-pooling outcome.

\FloatBarrier

\section{Conclusion}\label{sec:conclusion}

This paper developed a semantic signaling game for LLM-mediated
strategic communication. The model treats prompts and other high-level
instructions as semantic controls, the LLM as a stochastic language
channel, and receiver awareness as a type-dependent scoring system. This
structure connects token-level message generation with statistical
decision rules and equilibrium behavior.

The analysis shows how aggregate linguistic scores admit Gaussian
approximations, how receivers use likelihood-ratio thresholds, and how
sender incentives depend on the induced score distributions across
heterogeneous awareness types. The equilibrium results characterize when
type-consistent semantic controls are incentive compatible, while the
mechanism-design analysis identifies conditions for benign pooling and
population shaping.

The numerical study complements the theory by showing that the Gaussian
approximation tracks simulated score distributions, that greater
awareness reduces acceptance of stealth attacks, that adaptive senders
can exploit temporary blind regions during mindset evolution, and that
population and penalty interventions can restore benign pooling.

The framework also clarifies the role of systematic blindness. A
receiver may fail not because evidence is absent from the message, but
because informative features are absent from the receiver's effective
representation. Awareness shaping and mindset dynamics can therefore be
viewed as mechanisms for expanding the receiver's measurable feature
space. Future work can extend the model to richer multi-turn
interactions, adaptive prompting, empirical estimation of score
distributions from real LLM outputs, and mechanism design for teams of
human and artificial receivers.
\bibliographystyle{abbrv}
\bibliography{refs.bib}

@article{akata2025playing,
  title={Playing repeated games with large language models},
  author={Akata, Elif and Schulz, Lisa and Coda-Forno, Julian and Oh, Seong Joon and Bethge, Matthias and Schulz, Eric},
  journal={Nature Human Behaviour},
  volume={9},
  number={7},
  pages={1380--1390},
  year={2025},
  month={Jul}
}

@article{gallotta2024llm_games,
  title={Large language models and games: A survey and roadmap},
  author={Gallotta, Roberto and Todd, Graham and Zammit, Matthew and Earle, Sam and Liapis, Antonios and Togelius, Julian and Yannakakis, Georgios N.},
  journal={IEEE Transactions on Games},
  year={2024},
  month={Sep},
  note={Published online September 13, 2024}
}

@article{yang2026agentic_ai,
  title={Internet of Agentic AI: Incentive-Compatible Distributed Teaming and Workflow},
  author={Yang, Y. T. and Zhu, Quanyan},
  journal={arXiv preprint arXiv:2602.03145},
  year={2026},
  month={Feb},
  day={3}
}

@article{zhu2025cybersecurity,
  title={Game theory meets llm and agentic ai: Reimagining cybersecurity for the age of intelligent threats},
  author={Zhu, Quanyan},
  journal={arXiv preprint arXiv:2507.10621},
  year={2025},
  month={Jul},
  day={14}
}

@article{zhu2025llm_nash,
  title={Reasoning and behavioral equilibria in LLM-Nash games: From mindsets to actions},
  author={Zhu, Quanyan},
  journal={arXiv preprint arXiv:2507.08208},
  year={2025},
  month={Jul},
  day={10}
}

@article{banks1987equilibrium,
  title={Equilibrium selection in signaling games},
  author={Banks, Jeffrey S and Sobel, Joel},
  journal={Econometrica: Journal of the Econometric Society},
  pages={647--661},
  year={1987},
  publisher={JSTOR}
}

@article{cho1987signaling,
  title={Signaling Games and Stable Equilibria},
  author={Cho, In-Koo and Kreps, David M.},
  journal={The Quarterly Journal of Economics},
  volume={102},
  number={2},
  pages={179--221},
  year={1987},
  publisher={Oxford University Press}
}

@article{farrell1996cheap,
  title={Cheap Talk},
  author={Farrell, Joseph and Rabin, Matthew},
  journal={Journal of Economic Perspectives},
  volume={10},
  number={3},
  pages={103--118},
  year={1996}
}

@article{sobel2020lying,
  title={Lying and deception in games},
  author={Sobel, Joel},
  journal={Journal of Political Economy},
  volume={128},
  number={3},
  pages={907--947},
  year={2020},
  publisher={The University of Chicago Press Chicago, IL}
}

@article{aumann1999interactive,
  title={Interactive epistemology I: knowledge},
  author={Aumann, Robert J},
  journal={International Journal of Game Theory},
  volume={28},
  number={3},
  pages={263--300},
  year={1999},
  publisher={Springer}
}

@article{demichelis2008language,
  title={Language, meaning, and games: A model of communication, coordination, and evolution},
  author={Demichelis, Stefano and Weibull, J{\"o}rgen W},
  journal={American Economic Review},
  volume={98},
  number={4},
  pages={1292--1311},
  year={2008},
  publisher={American Economic Association}
}

@article{kamenica2011bayesian,
  title={Bayesian Persuasion},
  author={Kamenica, Emir and Gentzkow, Matthew},
  journal={American Economic Review},
  volume={101},
  number={6},
  pages={2590--2615},
  year={2011}
}

@article{bergemann2019information,
  title={Information Design: A Unified Perspective},
  author={Bergemann, Dirk and Morris, Stephen},
  journal={Journal of Economic Literature},
  volume={57},
  number={1},
  pages={44--95},
  year={2019}
}

@article{Hu2024GameTheoreticNP,
  author    = {Y. Hu and J. Chen and Q. Zhu},
  title     = {Game-Theoretic Neyman--Pearson Detection to Combat Strategic Evasion},
  journal   = {IEEE Transactions on Information Forensics and Security},
  year      = {2024},
  volume    = {20},
  pages     = {516--530},
  doi       = {}
}

@inproceedings{Zhang2018HypothesisTestingGame,
  author    = {T. Zhang and Q. Zhu},
  title     = {Hypothesis Testing Game for Cyber Deception},
  booktitle = {Proceedings of the International Conference on Decision and Game Theory for Security (GameSec)},
  year      = {2018},
  pages     = {540--555},
  publisher = {Springer},
  address   = {Cham},
  doi       = {}
}

@inproceedings{Li2023PriceTransparency,
  author    = {T. Li and Q. Zhu},
  title     = {On the Price of Transparency: A Comparison Between Overt Persuasion and Covert Signaling},
  booktitle = {Proceedings of the 62nd IEEE Conference on Decision and Control (CDC)},
  year      = {2023},
  pages     = {4267--4272},
  publisher = {IEEE},
  doi       = {}
}

@article{Pawlick2019LeakyDeception,
  author    = {J. Pawlick and E. Colbert and Q. Zhu},
  title     = {Modeling and Analysis of Leaky Deception Using Signaling Games with Evidence},
  journal   = {IEEE Transactions on Information Forensics and Security},
  year      = {2019},
  volume    = {14},
  number    = {7},
  pages     = {1871--1886},
  doi       = {}
}

@book{hall2014martingale,
  title={Martingale limit theory and its application},
  author={Hall, Peter and Heyde, Christopher C},
  year={2014},
  publisher={Academic press}
}

@book{wittgenstein1922tractatus,
  title={Tractatus Logico-Philosophicus},
  author={Wittgenstein, Ludwig},
  year={1922},
  publisher={Routledge and Kegan Paul},
  address={London},
  note={Originally published in German as "Logisch-Philosophische Abhandlung"}
}

@book{simon1957models,
  title={Models of Man: Social and Rational; Mathematical Essays on Rational Human Behavior in a Social Setting},
  author={Simon, Herbert A.},
  year={1957},
  publisher={Wiley},
  address={New York}
}

@book{cover1999elements,
  title={Elements of Information Theory},
  author={Cover, Thomas M. and Thomas, Joy A.},
  edition={2nd},
  year={2006},
  publisher={Wiley},
  address={New York}
}

@article{white1982maximum,
  title={Maximum Likelihood Estimation of Misspecified Models},
  author={White, Halbert},
  journal={Econometrica},
  volume={50},
  number={1},
  pages={1--25},
  year={1982}
}

@article{zhu2025llm,
  title={LLM-Stackelberg Games: Conjectural Reasoning Equilibria and Their Applications to Spearphishing},
  author={Zhu, Quanyan},
  journal={arXiv preprint arXiv:2507.09407},
  year={2025}
}

@inproceedings{zhu2025generative,
  title={Generative-Conjectural LLM Equilibrium for Agentic AI Deception with Applications to Spearphishing},
  author={Zhu, Quanyan},
  booktitle={International Conference on Game Theory and AI for Security},
  pages={356--375},
  year={2025},
  organization={Springer}
}

@inproceedings{ouyang2022training,
  title={Training Language Models to Follow Instructions with Human Feedback},
  author={Ouyang, Long and Wu, Jeffrey and Jiang, Xu and Almeida, Diogo and Wainwright, Carroll L. and Mishkin, Pamela and Zhang, Chong and Agarwal, Sandhini and Slama, Katarina and Ray, Alex and Schulman, John and Hilton, Jacob and Kelton, Fraser and Miller, Luke and Simens, Maddie and Askell, Amanda and Welinder, Peter and Christiano, Paul and Leike, Jan and Lowe, Ryan},
  booktitle={Advances in Neural Information Processing Systems},
  volume={35},
  pages={27730--27744},
  year={2022}
}

@inproceedings{park2023generative,
  title={Generative Agents: Interactive Simulacra of Human Behavior},
  author={Park, Joon Sung and O'Brien, Joseph C. and Cai, Carrie J. and Morris, Meredith Ringel and Liang, Percy and Bernstein, Michael S.},
  booktitle={Proceedings of the 36th Annual ACM Symposium on User Interface Software and Technology},
  pages={1--22},
  year={2023},
  publisher={ACM}
}

@article{huang2022advert,
  title={Advert: an adaptive and data-driven attention enhancement mechanism for phishing prevention},
  author={Huang, Linan and Jia, Shumeng and Balcetis, Emily and Zhu, Quanyan},
  journal={IEEE Transactions on Information Forensics and Security},
  volume={17},
  pages={2585--2597},
  year={2022},
  publisher={IEEE}
}

@article{pawlick2017phishing,
  title={Phishing for phools in the internet of things: Modeling one-to-many deception using poisson signaling games},
  author={Pawlick, Jeffrey and Zhu, Quanyan},
  journal={arXiv preprint arXiv:1703.05234},
  year={2017}
}

@article{jagatic2007social,
  title={Social Phishing},
  author={Jagatic, Tom N. and Johnson, Nathaniel A. and Jakobsson, Markus and Menczer, Filippo},
  journal={Communications of the ACM},
  volume={50},
  number={10},
  pages={94--100},
  year={2007},
  publisher={ACM}
}

@book{huang2023cognitive,
  title={Cognitive Security: A System-Scientific Approach},
  author={Huang, Ling and Zhu, Quanyan},
  publisher={Springer Nature},
  year={2023},
  month={Jun},
  day={2}
}

@article{spence1973job,
  title={Job Market Signaling},
  author={Spence, Michael},
  journal={The Quarterly Journal of Economics},
  volume={87},
  number={3},
  pages={355--374},
  year={1973},
  publisher={Oxford University Press}
}

@article{crawford1982strategic,
  title={Strategic Information Transmission},
  author={Crawford, Vincent P. and Sobel, Joel},
  journal={Econometrica},
  volume={50},
  number={6},
  pages={1431--1451},
  year={1982}
}

@book{lewis1969convention,
  title={Convention: A Philosophical Study},
  author={Lewis, David},
  publisher={Harvard University Press},
  address={Cambridge, MA},
  year={1969}
}

@inproceedings{vaswani2017attention,
  title={Attention Is All You Need},
  author={Vaswani, Ashish and Shazeer, Noam and Parmar, Niki and Uszkoreit, Jakob and Jones, Llion and Gomez, Aidan N. and Kaiser, Lukasz and Polosukhin, Illia},
  booktitle={Advances in Neural Information Processing Systems},
  volume={30},
  year={2017}
}

@inproceedings{brown2020language,
  title={Language Models Are Few-Shot Learners},
  author={Brown, Tom B. and Mann, Benjamin and Ryder, Nick and Subbiah, Melanie and Kaplan, Jared and Dhariwal, Prafulla and Neelakantan, Arvind and Shyam, Pranav and Sastry, Girish and Askell, Amanda and Agarwal, Sandhini and Herbert-Voss, Ariel and Krueger, Gretchen and Henighan, Tom and Child, Rewon and Ramesh, Aditya and Ziegler, Daniel M. and Wu, Jeffrey and Winter, Clemens and Hesse, Christopher and Chen, Mark and Sigler, Eric and Litwin, Mateusz and Gray, Scott and Chess, Benjamin and Clark, Jack and Berner, Christopher and McCandlish, Sam and Radford, Alec and Sutskever, Ilya and Amodei, Dario},
  booktitle={Advances in Neural Information Processing Systems},
  volume={33},
  pages={1877--1901},
  year={2020}
}

@article{bommasani2021opportunities,
  title={On the Opportunities and Risks of Foundation Models},
  author={Bommasani, Rishi and Hudson, Drew A. and Adeli, Ehsan and Altman, Russ and Arora, Simran and von Arx, Sydney and Bernstein, Michael S. and Bohg, Jeannette and Bosselut, Antoine and Brunskill, Emma and Brynjolfsson, Erik and Buch, Shyamal and Card, Dallas and Castellon, Rodrigo and Chatterji, Niladri and Chen, Annie and Creel, Kathleen and Davis, Jared Quincy and Demszky, Dora and Donahue, Chris and Doumbouya, Moussa and Durmus, Esin and Ermon, Stefano and Etchemendy, John and Ethayarajh, Kawin and Fei-Fei, Li and Finn, Chelsea and Gale, Trevor and Gillespie, Lauren and Goel, Karan and Goodman, Noah and Grossman, Shelby and Guha, Neel and Hashimoto, Tatsunori and Henderson, Peter and Hewitt, John and Ho, Daniel E. and Hong, Jenny and Hsu, Kyle and Huang, Jing and Icard, Thomas and Jain, Saahil and Jurafsky, Dan and Kalluri, Pratyusha and Karamcheti, Siddharth and Keeling, Geoff and Khani, Fereshte and Khattab, Omar and Koh, Pang Wei and Krass, Mark and Krishna, Ranjay and Kuditipudi, Rohith and Kumar, Ananya and Ladhak, Faisal and Lee, Mina and Lee, Tony and Leskovec, Jure and Levent, Isabelle and Li, Xiang Lisa and Li, Xuechen and Ma, Tengyu and Malik, Ali and Manning, Christopher D. and Mirchandani, Suvir and Mitchell, Eric and Munyikwa, Zanele and Nair, Suraj and Narayan, Avanika and Narayanan, Deepak and Newman, Ben and Nie, Allen and Niebles, Juan Carlos and Nilforoshan, Hamed and Nyarko, Julian and Ogut, Giray and Orr, Laurel and Papadimitriou, Isabel and Park, Joon Sung and Piech, Chris and Portelance, Eva and Potts, Christopher and Raghunathan, Aditi and Reich, Rob and Ren, Hongyu and Rong, Frieda and Roohani, Yusuf and Ruiz, Camilo and Ryan, Jack and R{\'e}, Christopher and Sadigh, Dorsa and Sagawa, Shiori and Santhanam, Keshav and Shih, Andy and Srinivasan, Krishnan and Tamkin, Alex and Taori, Rohan and Thomas, Armin W. and Tram{\`e}r, Florian and Wang, Rose E. and Wang, William and Wu, Bohan and Wu, Jiajun and Wu, Yuhuai and Xie, Sang Michael and Yasunaga, Michihiro and You, Jiaxuan and Zaharia, Matei and Zhang, Michael and Zhang, Tianyi and Zhang, Xikun and Zhang, Yuhui and Zheng, Lucia and Zhou, Kaitlyn and Liang, Percy},
  journal={arXiv preprint arXiv:2108.07258},
  year={2021}
}

@misc{cyb3rmik3_phishing_keywords_2023,
  title={{Hunting-Lists}: phishing-keywords.csv},
  author={{cyb3rmik3}},
  howpublished={\url{https://github.com/cyb3rmik3/Hunting-Lists/commit/d1d36d7c7ee6c4e8bc491e8ae022e82c135e76da}},
  note={GitHub commit d1d36d7c7ee6c4e8bc491e8ae022e82c135e76da, accessed June 16, 2026},
  year={2023}
}


\end{document}